\documentclass[10pt,conference,a4paper]{IEEEtran}
\usepackage{amsmath}
\usepackage{color}\usepackage[pdftex]{graphicx}
\usepackage{amssymb}
\usepackage{amsthm}
\usepackage{epstopdf}
\newtheorem{thm}{Theorem}

\newtheorem{lem}{Lemma}

\title{ Finite Blocklength Rates over a Fading Channel with CSIT and CSIR}
\author{Deekshith P K  and Vinod Sharma \\ {ECE Dept., Indian Institute of Science, Bangalore, India.}\\{\hspace{-.6cm}Email: \{deeks, vinod\}@ece.iisc.ernet.in}}

\begin{document}
\maketitle
\begin{abstract}
In this work, we obtain lower and upper bounds on the maximal {transmission} rate at a given codeword length $n$, average probability of error $\epsilon$ and power constraint $\bar{P}$, over a finite valued, block fading additive white Gaussian noise (AWGN) channel with channel state information (CSI) at the transmitter {and the receiver}. {These bounds characterize} deviation of the finite blocklength coding rates from the channel capacity which is in turn achieved by the \textit{water filling} power allocation across time. {The bounds obtained also} characterize the rate enhancement possible due to the {CSI} at the transmitter in the finite blocklength regime. The results are further elucidated via numerical examples.
\end{abstract}
\noindent
\section{Introduction}
Next generation cellular {networks ought} to handle \textit{mission critical} data with delay requirements far more stringent than that in present day cellular networks \cite{agiwal2016next}. Refined engineering insights to build such delay critical systems can be obtained using  the analytical methods pioneered  in \cite{strassen1962asymptotische}, \cite{polyanskiy2010channel}, \cite{hayashi2009information}. In this work, we characterize data rate enhancement in a wireless system with delay constraints by means of \textit{power adaptation}, when the transmitter has certain \textit{side information} about the channel. We restrict our attention to the delay incurred at the physical layer in sending the codeword to the receiver.   
\par In a cellular system, if the \textit{instantaneous channel gain} can be fed back to the transmitter, the transmitter can use this knowledge to perform power control so as to increase the \textit{overall data rate}. In particular, under the assumption of perfect CSI at the transmitter (CSIT) and the receiver (CSIR), the optimal power allocation with no delay constraints over a flat fading AWGN channel has the well known interpretation of \textit{water filling} in time {(\cite{goldsmith1997capacity})}.
\par With delay constraints imposed at the physical layer, the traditional approach to study rate enhancement due to the knowledge of CSIT is to first characterize either the delay limited, outage or average capacity (see \cite{el2011network} for details) and then obtain a power allocation strategy that maximizes the required one of these quantities. In this regard, \cite{caire1999optimum} obtains the optimal power allocation that maximises the outage capacity under the assumption of \textit{non-causal} CSIT. In \cite{negi2002delay}, the authors obtain the optimal power allocation scheme maximizing the average capacity with causal CSIT. Nonetheless, the above mentioned schemes do not provide a realistic metric to evaluate the performance of actual delay sensitive systems. This is because, the various notions of capacity used therein are inherently asymptotic.  
\par In this work, we provide lower and upper bounds on the \textit{maximal channel coding rate} over a block fading AWGN channel with CSIT and CSIR in \textit{finite blocklength regime} under two kinds of constraints on the transmitted codewords. Consequently, we characterize the rate enhancement possible due to power adaptation at the transmitter. Our assumption of perfect CSIT is idealistic. Nevertheless, rates obtained under this assumption provide upper bounds for rates achievable without CSIT or with imperfect CSIT {and is commonly made in literature}. Also, {knowledge of the power control strategies suitable for delay constrained systems {sheds} insights into how system energy is to be used in such systems.} Efficient usage of system energy can be beneficial for energy constrained transmitters ought to become prominent { in future wireless networks \cite{mahapatra2016energy}}.  
\par An overview of the preceeding works is presented next. A scalar coherent fading channel with stationary fading (generalization of block fading)  \textit{without} CSIT is considered in \cite{polyanskiy2011scalar} and the \textit{dispersion} term is characterized. In \cite{yang2015optimum}, the authors show that the (second order) optimal power allocation scheme over a \textit{quasi static} fading channel with CSIT and CSIR is \textit{truncated channel inversion}. {The quasi static fading model is a special case of the block fading model that we consider. However, our bounds involve asymptotic terms, asymptotic in the number of blocks and is derived under the assumption of a finite valued fading process.  A MIMO Rayleigh block fading channel with no CSIT and CSIR is considered in \cite{durisi2016short} and achievability and converse bounds are derived for the \textit{short packet communication} regime. In \cite{tomamichel2014second}, {the} authors consider a \textit{discrete memoryless channel} with {non-causal} CSIT and CSIR, and obtains the second order coding rates under general assumptions on the state process. In contrast, we consider a cost constrained setting with real valued channel inputs. {This renders a direct translation of the techniques used therein infeasible. A high-SNR normal approximation of the maximal coding rate over a block fading Rayleigh channel without CSIT and CSIR is obtained in \cite{lancho2017high}.
\par Our contribution is a finer characterization of the delay limited performance of a wireless link with channel state feedback, under two kinds of power constraints  the wireless transmitters are normally subjected to. The bounds obtained (on the maximal coding rate)  characterize the rate enhancement due to power adaptation for a given codeword length and error probability. In deriving these bounds, the CSIT assumption makes the analysis involved and non trivial.  In particular, in obtaining the upper bounds, the dependence of the channel input on the fading states makes the  corresponding optimization problems difficult to solve. To circumvent this, we derive alternate bounds utilizing the properties of asymptotically optimal power allocation scheme, viz, the water filling scheme. 
\par {The paper is organized as follows.} In Section \ref{mod_notn}, we introduce the system model and notation. We {provide}  lower  bounds on the maximal channel coding rate in Section \ref{lower_bnd}. Next, in Section \ref{sec_up_bnd}, we  provide upper bounds for the maximal coding rate. In Section \ref{num_ex}, we compare the bounds numerically and exemplify the utility of the bounds derived. We conclude in Section \ref{con}. Proofs are delegated to the appendices.    
\section{Model and Notation}
\label{mod_notn}
We consider a point to point {discrete time, memoryless} block fading channel subject to AWGN noise with density $\mathcal{N}(0,\sigma_N^2)$, {where $\mathcal{N}(a,b)$ denotes Gaussian density with mean $a$ and variance $b$.} The noise is  independent and identically distributed (i.i.d.) across channel uses. Let $T_c$ denote the \textit{channel coherence time} which is the duration over which the gain of the underlying physical channel remains constant. Let $D$ (in appropriate time units) denote the delay constraint imposed on the communication. Then, $B=\lceil\frac{D}{T_c}\rceil$ denotes the number of blocks over which the communication spans {(where $\lceil{x}\rceil$ denotes the smallest integer $\geq x$)}. Let $n_c$ denote the number of times the channel is used within a block. Then, the number of channel uses for the whole {of communication}, or equivalently, the codeword length $n=Bn_c$. 
\par The channel gain or the fading coefficient in block $b$ is denoted as $H_b$, which is a random variable taking values in a finite set $\mathcal{H}=\{\eta_1,\hdots,\eta_{|\mathcal{H}|}\}$ such that $\min\{\eta_i:{1\leq i\leq |\mathcal{H}|}\}>0$. Here $|\mathcal{H}|$ (the cardinality of the set) denotes the number of fading states. Let $q_i$ denote the probability of $H_b$ taking value $\eta_i$ and $q_i>0$. The channel gains are i.i.d. across blocks and is independent of the additive noise process. {The instantaneous channel gains are assumed to be known to the transmitter  as well as the receiver and the transmitter gets to know them only causally (we refer to it as the full CSIT and CSIR assumption).}  
\par Let $X_{(b-1)n_c+k}$ denote the channel input corresponding to {the} $k^{\text{th}}$ channel use in {the} $b^{\text{th}}$ block, where $k \in \{1,\hdots,n_c\}$, $b \in \{1,\hdots,B\}$. For convenience, from here on, $[bk]\equiv (b-1)n_c+k$.  Let $Z_{[bk]}$ (distributed as $\mathcal{N}(0,\sigma_N^2)$) and $Y_{[bk]}$ denote the corresponding noise variable and the channel output, respectively. Then, $Y_{[bk]}=H_{b}X_{[bk]}+Z_{[bk]}.$ If the delay $D$ tends to infinity (and hence, the number of blocks $B$ tends to infinity), it is well known that the channel capacity is given by {(\cite{goldsmith1997capacity})}
\begin{equation}
\label{gold_eqn}
\mathbf{C}(\bar{P})\triangleq\mathbb{E}_{\text{H}}\Bigg[\frac{1}{2}\log\Big(1+\frac{H^2\mathcal{P}_{\text{WF}}(H)}{\sigma_N^2}\Big)\Bigg],
\end{equation}
 where $\mathcal{P}_{\text{WF}}(H)\triangleq \Big(\lambda-\frac{\sigma_N^2}{H^2}\Big)^+,$ $(.)^+=\max(0,.)$, $\mathbb{E}_{\text{H}}[.]$ denotes the expectation with respect to the distribution of the random variable $H$ and $\lambda$ is obtained by solving the equation $\mathbb{E}_{\text{H}}\Big[\mathcal{P}_{\text{WF}}(H)\Big]=\bar{P}.$ Here, $\mathcal{P}_{\text{WF}}(.)$ {is the water filling solution with average power constraint $\bar{P}$}.
\par {We make note of certain notations that we use throughout}. Let $ C(x)\triangleq \frac{1}{2}\log\big(1+\frac{x}{\sigma_N^2}\big)$, $\mathcal{L}(x)\triangleq \frac {x}{\sigma_N^2+x}$, $V(x)\triangleq \frac{1}{2}\big[1-\big(1-\mathcal{L}(x)\big)^2\big]$. A set of positive integers $\{1,2,\hdots,N\}$ is denoted as $[1:N]$. Wherever required, we denote a sum of $n$ numbers $\big\{a_k,~k\in[1:n]\big\}$, $\sum_{k=1}^{n}a_k$, as $S_n(a)$. We choose to represent the channel input and output vectors conveniently as {a} collection of $B$  vectors, {each}  of length $n_c$. Thus, the channel input $\mathbf{x} \equiv (x_1,\hdots,x_n)=(\underline{x}_1,\hdots,\underline{x}_{{B}})$, where $\underline{x}_b=(x_{[b1]},\hdots,x_{[bn_c]}),~b \in [1:B]$. Similarly, we have the noise vector $\mathbf{z}=(\underline{z}_1,\hdots,\underline{z}_{B})$ and the channel output vector $\mathbf{y}=(\underline{y}_1,\hdots,\underline{y}_{B})$. Also, the vector of channel fading gains $\mathbf{h}= (h_1,\hdots,h_B)$ is a collection of $B$ scalars.  Corresponding random variables are denoted as $\mathbf{X},~\mathbf{Z},~\mathbf{Y}$ and $\mathbf{H}$. Let $\eta_{\min} \triangleq \min\{\eta_i:{i\in[1:|\mathcal{H}|]}\}$, $q_{\min} \triangleq \min\{q_i:{i\in[1:|\mathcal{H}|]}\}$. Similarly, the corresponding maximum values are denoted as $\eta_{\max}$ and $q_{\max}$, respectively. The Cartesian product of two sets $S_1$, $S_2$ is denoted as $S_1\times S_2$ and $n$ fold Cartesian product of sets $S_1,\hdots,S_n$ is denoted as $S^n$. The set of integers is denoted as $\mathbb{Z}$ and {the} set of positive integers as $\mathbb{Z}_+$. {The real line is denoted as $\mathbb{R}$, positive real line as $\mathbb{R}_+$, the $n$ dimensional Euclidean space by $\mathbb{R}^n$ and $\mathbb{R}_+^n\triangleq \mathbb{R}_+\times\hdots\times\mathbb{R}_+$ ($n$ times). Given vectors $\mathbf{x},~\mathbf{y} \in \mathbb{R}^n$, $||\mathbf{x}||$ denotes the Euclidean norm of $\mathbf{x}$ and $\langle \mathbf{x},\mathbf{y}\rangle$ denotes the inner product between $\mathbf{x}$ and $\mathbf{y}$. The variance of a random variable $X$ is denoted as $\mathbb{V}[X]$. The function $\Phi(.)$ denotes the cumulative distribution function (cdf) of a standard Gaussian random variable, $\phi(.)$ denotes the corresponding density function and $\Phi^{-1}(.)$ denotes its inverse cdf. The notation $f_n=o(g_n)$ is equivalent to  $\displaystyle \lim_{n \rightarrow \infty}f_n/g_n=0$. Also, $f_n=O(g_n)$ is equivalent to $|f_n|\leq K|g_n|$, for some constant $K$, for all $n$ sufficiently large. The notation $\overline{\lim}\equiv \limsup$. {To indicate relations that hold almost surely, we use the abbreviation a.s.}. We use the notation $\stackrel{D}{=}$ to mean equivalence in distribution. {We denote the indicator function of an event as $\mathbf{1}_A$. The exponential function is denoted as $\text{exp}(.)$. For any set $A$, $A^c$ denotes its complement. All logarithms are taken to the natural base.}  

\section{Maximal Coding Rate: Lower bounds}
\label{lower_bnd}
In this section, we will obtain  lower bounds on the maximal coding rate for a given codeword length and average probability of error, under two different kinds of power constraints on the transmitted codewords. We have the following definitions. Let $\mathcal{M} \equiv [1:M]$ denote the message set. Let $S$ be a uniformly distributed random variable, corresponding to the message transmitted, taking values in $\mathcal{M}$. Given $S \in \mathcal{M}$ and the fading coefficients $ (H_1,\hdots,H_b)$, at $k^{\text{th}}$ instance of channel use in block $b$, the output of the encoder $\zeta_{[bk]}:\mathcal{M}\times\mathcal{H}^b \rightarrow \mathbb{R}$  is denoted as  $X_{[bk]}$. The decoder $\psi:\mathbb{R}^n\times\mathcal{H}^B \rightarrow \mathcal{M}$, on obtaining the $(\mathbf{Y},\mathbf{H})$ pair, outputs an estimate of the message $\hat{S}$. The encoding and decoding is done so that the average probability of error $P\big(\psi(\mathbf{Y},\mathbf{H}) \neq S\big)\leq \epsilon$, where $\epsilon$ is prefixed. Throughout this work, we adhere to the average probability of error formalism.
\par In this work, we consider two types of power constraints: 
\begin{equation}
\label{ST_avg_PC}
\hspace{-5pt}
\textbf{ST}: ~\sum_{b=1}^{B}\sum_{k=1}^{n_c}\zeta_{[bk]}^2(m,H^b) \stackrel{\text{a.s}}{\leq} Bn_c\bar{P},~\forall m \in \mathcal{M}.
\end{equation}
\begin{equation}
\label{LT_avg_PC}
\textbf{LT}: ~\mathbb{E}_{\mathbf{H}}\Bigg[\sum_{b=1}^{B}\sum_{k=1}^{n_c}\zeta_{[bk]}^2(m,H^b)\Bigg] \leq Bn_c\bar{P},~\forall m \in \mathcal{M}.
\end{equation}
 The constraint in equation (\ref{ST_avg_PC}) is referred to as the \textit{short term power constraint} ($\mathbf{ST}$) and  that in (\ref{LT_avg_PC}),  as the  \textit{long term power constraint} ($\mathbf{LT}$). Studying a communication system  under the $\mathbf{ST}$ constraint is motivated by the peak power limitations of the circuitry involved. Whereas, imposing the $\mathbf{LT}$ constraint captures the requirement of power utilization efficiency  in a communication device (for instance, battery powered mobile radio transmitters). In addition, in a wireless communication setting, in studying systems which allocate resources (e.g., rate, power) dynamically, the $\mathbf{LT}$ constraint is a natural metric to consider. Though in reality these constraints are simultaneously present, in this work, we study them in isolation. For each of the above constraint, the goal is to characterize the maximum size (or cardinality) of the codebook with block length $n$ and average probability of error  $\epsilon$,  denoted as   $M^*(n,\epsilon,\bar{P})\equiv M^*$. The maximal coding rate $R^*(n,\epsilon,\bar{P})\equiv R^*=n^{-1}{\log M^*}$.
 \par   Our first result gives a lower bound on $\log M^*$ under $\mathbf{ST}$ constraint.
\begin{thm}
\label{Th_ST_LB2}
For a block fading channel with input subject to a short term power constraint $\bar{P}$ and average probability of error $\epsilon$, the maximal codebook size $ M^*$ satisfies
\begin{equation}
\label{Result2}
{\log M^* \geq n\mathbf{C}(\bar{P})-\sqrt{\frac{n}{2}}+\sqrt{n{V_{\text{BF}}(\bar{P})}}\Phi^{-1}(\epsilon)+o(\sqrt{n}).}
\end{equation}
Here, with $G \triangleq H^2\mathcal{P}_{\text{WF}}(H) $, 
\begin{flalign*}
\hspace{-20pt}
V_{\text{BF}}(\bar{P}) \triangleq \mathbb{E}_{\text{G}}\big[V(G)\big]+n_c\mathbb{V} \big[C(G)\big]+\frac{1}{2}\mathbb{V}\big[\mathcal{L}(G)\big].
\end{flalign*}
  \end{thm}
\begin{proof}
See Appendix A.
\end{proof}
 The following result provides a lower bound on $\log M^*$ subject to $\mathbf{LT}$ constraint.
\begin{thm}
\label{Th_LT_LB1}
For a block fading channel with input subject to a long term power constraint $\bar{P}$ and average probability of error $\epsilon$, the maximal codebook size $M^*$ at a given blocklength $n$ satisfies 
\begin{equation}
\label{Result1}
\log M^* \geq n\mathbf{C}(\bar{P})+\sqrt{n{V_{\text{BF}}}(\bar{P})}\Phi^{-1}(\epsilon)+o(\sqrt{n}),
\end{equation}
where, the notation is as in Theorem \ref{Th_ST_LB2}.
 \end{thm}
\begin{proof}
See Appendix B.
\end{proof}
\section{Maximal Coding Rate: Upper bounds}
\label{sec_up_bnd}
In this section, we provide upper bounds on the maximal coding rate under the $\mathbf{ST}$ and $\mathbf{LT}$ constraints. In deriving the upper bounds, we assume that CSIT is known non-causally. First, we obtain an upper bound for the $\mathbf{ST}$ case. {Next}, we proceed to provide an upper bound for the $\mathbf{LT}$ case. 
\begin{thm}
\label{Th_ST_UB1}
For a block fading channel with input subject to a short term power constraint $\bar{P}$, average probability of error $0<\epsilon<\frac{1}{2}$, the maximal codebook size $M^*$ satisfies 
\begin{equation}
{\log M}^*\leq n\mathbf{C}(\bar{P})+\sqrt{nV_{\text{BF}}'(\bar{P})}\Phi^{-1}(\epsilon)+O({\log n}),
\label{Result3}
\end{equation}
{where, with $G$ as in Theorem \ref{Th_ST_LB2} and $\lambda$ as in (\ref{gold_eqn}), $V_{\text{BF}}'(\bar{P}) \triangleq \mathbb{E}_{\text{G}}[V(G)]+\mathbb{V}\Big[n_cC(G)+\frac{\bar{P}}{2\lambda}-\frac{\mathcal{L}(G)}{2}\Big].$  }
\end{thm}
\begin{proof}
See Appendix C.
\end{proof}
\begin{thm}
\label{Th_LT_UB2}
For a block fading channel with input subject to a long term power constraint $\bar{P}$ and average probability of error $0<\epsilon<\frac{1}{2}$, the maximal codebook size $M^*$ satisfies 
\begin{equation}
\label{Result4}
\log  M^* \leq n\mathbf{C}(\bar{P})+\sqrt{\frac{n}{4\lambda^2}}+\sqrt{{nV_{\text{BF}}'(\bar{P})}}\Phi^{-1}(\epsilon)+O({\log n}),
\end{equation} 
 where, $\lambda$ is as in (\ref{gold_eqn}) and $V_{\text{BF}}'(\bar{P})$ is as in Theorem \ref{Th_ST_UB1}. 
\end{thm}
\begin{proof}
See Appendix D.
\end{proof}
\section{Numerical Examples}

In this section, we compare numerically, the various bounds obtained thus far. To that end, we need a refined characterization of the lower bounds. In particular, the $o(\sqrt{n})$ terms need to be explicitly identified. Terms contributing to the $o(\sqrt{n})$ expression in the lower bound for the $\mathbf{LT}$ constraint are same as that in the no CSIT case with $\mathbf{ST}$ constraint (\cite{polyanskiy2011scalar}). Hence, the terms can be identified by revisiting the analysis in \cite{polyanskiy2011scalar} and computing the constants therein, explicitly for the i.i.d. block fading case. Doing so, we obtain $o(\sqrt{n})=\frac{\log n}{2}-{\sqrt{n^{1-\beta}}}+O(n^\beta),$ for some small $\beta>0$.  For our computation, we fix $\beta=0.01.$ 
\par Next, consider the lower bound for the $\mathbf{ST}$ constraint case. Note that we are interested in error probability $\epsilon<1/2$ for which, $\Phi^{-1}(\epsilon)<0$. Under this assumption, the right hand side of (\ref{ST_M_bnd}) can be readily simplified using Taylor's theorem and the fact that $V_{\text{BF}}^t(.)$ therein, is a monotonically increasing function. By virtue of choice of $\delta_B$ (as mentioned in Appendix A), $\Phi^{-1}(\epsilon-\delta_B')=\Phi^{-1}(\epsilon)+O(1/n^{1.5})$. Thus, we obtain $\log M^*\geq n\mathbf{C}(\bar{P})-\sqrt{n/2}+\sqrt{nV_{\text{BF}}(\bar{P})}\Phi^{-1}(\epsilon)+o(\sqrt{n}),$ where the $o(\sqrt{n})$ term is same as in the $\mathbf{LT}$ case. The coefficient of $\log n/n$ in the upper bound for the $\mathbf{ST}$ and $\mathbf{LT}$ case is $|\mathcal{H}|/2$. We compute the terms excluding the $O(1/n^{1-\beta})$ terms in the lower bound and $O(1/n)$ term in the upper bound. Thus, the approximate bound that we compute is akin to the normal approximation for AWGN channels (\cite{polyanskiy2010channel}). With this, we proceed to compute the bounds. 
\par We consider the following \textit{discrete version} of the Rayleigh distribution: fix $\eta_{0}=0.1,$ $|\mathcal{H}|=10$, $~\eta_{|\mathcal{H}|-1}=4.1$  and $\Delta=(\eta_{|\mathcal{H}|-1}-\eta_{0})/(|\mathcal{H}|-1)$. Define $\eta_i=\eta_{i-1}+i\Delta$, for $i\in[1:|\mathcal{H}|-2]$ and $q_i=P(\eta_i\leq H_R\leq\eta_{i+1} ),$ for $i\in[0:|\mathcal{H}|-1]$ and $q_{|\mathcal{H}|-1}=P(H_R\geq \eta_{|\mathcal{H}|-1})$ where $H_R\sim$  Rayleigh distribution with parameter unity. Fix $\sigma_N^2=1$ and $\epsilon=0.01$ and $n_c=1$ (which corresponds to the \textit{fast fading} case). By fixing $\bar{P}=5$dB, we plot the convergence of various bounds to the channel capacity $\mathbf{C}(\bar{P})=0.6502$ nats/channel use (equation (\ref{gold_eqn})), as $B$ increases, in Figure \ref{fig_R_B}. The acronyms LB and UB refer to lower bound and upper bound. {We have also plotted the rates with no CSIT and $\mathbf{ST}$ constraint.
\begin{figure}[h]
\begin{center}
\hspace{-20pt}
\includegraphics[scale=0.45]{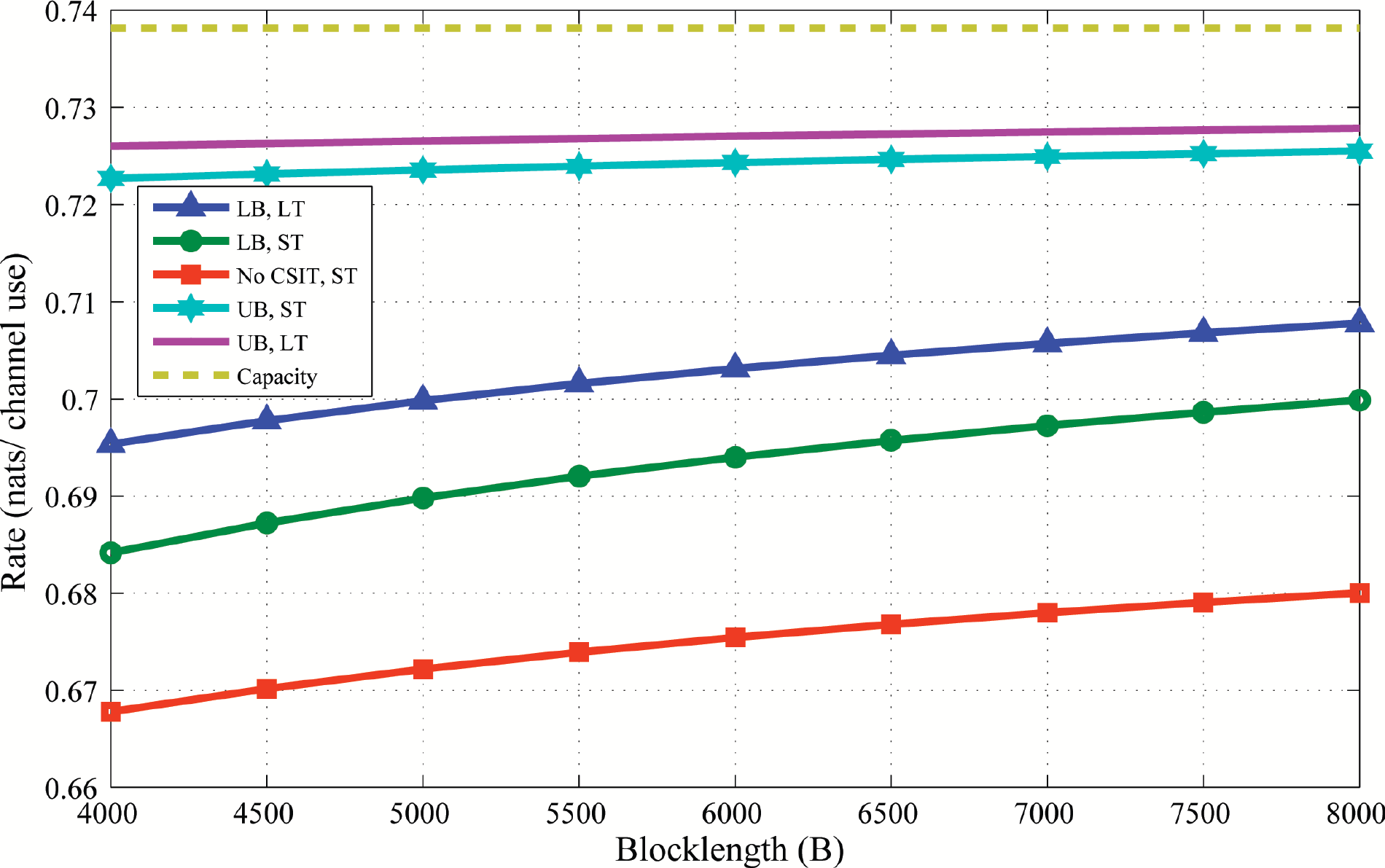}
\caption{Rate versus blocklength.} \label{fig_R_B}
\end{center}
\end{figure}
\begin{figure}[h]
\begin{center}
\hspace{-20pt}
\includegraphics[scale=0.45]{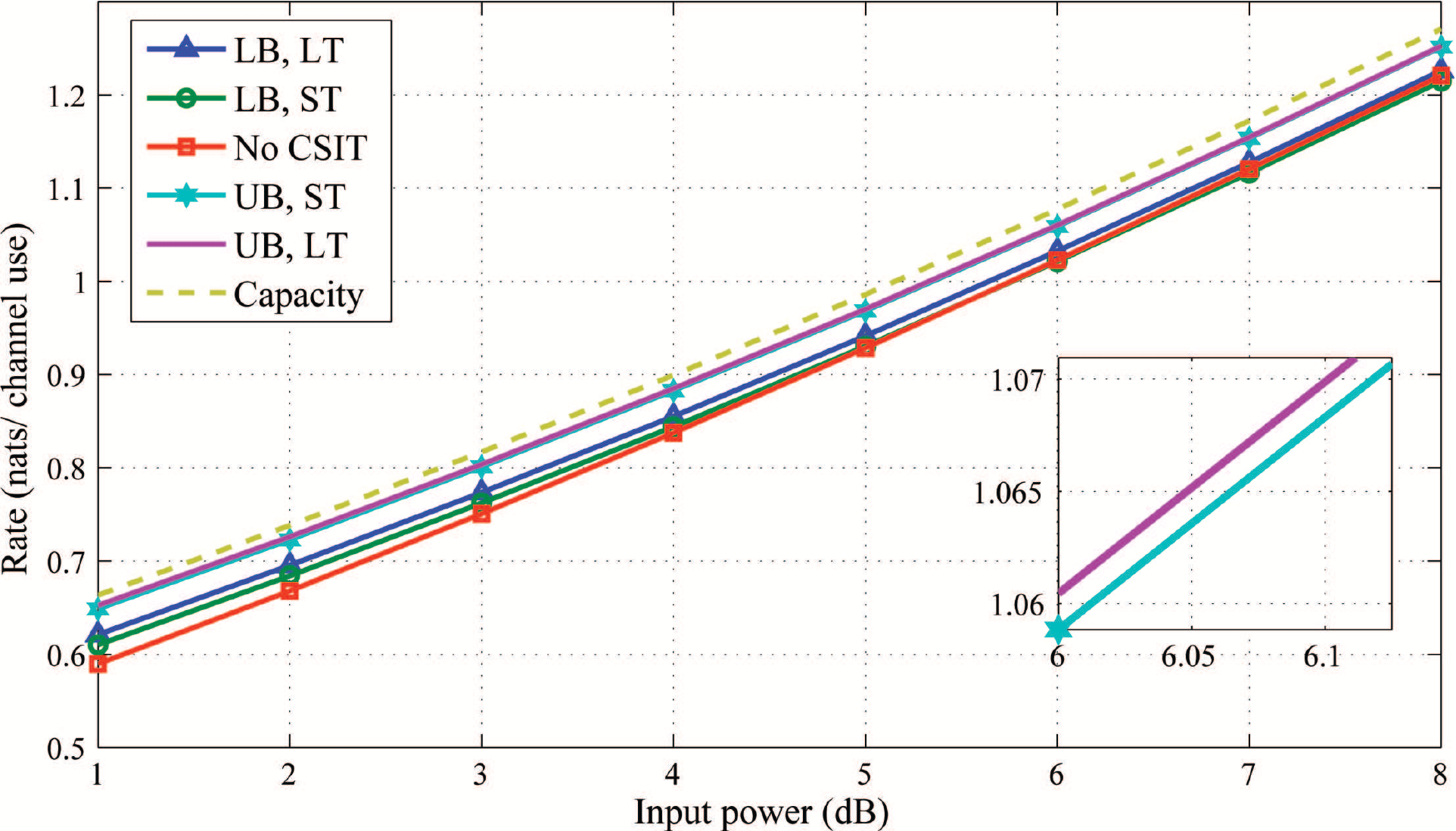}
\caption{Rate versus input power.} \label{fig_R_P}
\end{center}
\end{figure}
\par In Figure \ref{fig_R_P}, we fix $n_c=1,~B=4000$ and compare the various bounds for different values of $\bar{P}$.  By comparing the lower bound under the $\mathbf{ST}$ constraint with the rate for no CSIT case under $\mathbf{ST}$ constraint, the rate enhancement possible due to the knowledge of CSIT (via power control) is observed. 
\label{num_ex}
\section{Conclusion}
In this paper, we {have} obtained upper and lower bounds for the maximal coding rate over a block fading channel {with short term and long term average} power constraints on the transmitted codeword. The bounds obtained shed light {on} the rate enhancement possible due to the availability of CSIT. The bounds also characterize the performance of water filling power allocation in the finite block length regime. 
\label{con}
\section*{Acknowledgement}
The authors would like to thank Gautam Konchady Shenoy, ECE Dept., IISc, Bangalore and Wei Yang, Princeton University, for valuable discussions and useful comments.

\section*{Appendix A}
 \subsection*{Lower bound: $\mathbf{ST}$ constraint}
 Let $\mathcal{F}^{(n)}=\{\mathbf{x}' \in \mathbb{R}^{n}:||\mathbf{x}'||^2\ = n\}.$ Let $\mathcal{C}(M,n,\epsilon,1)$ be a codebook with $M$ codewords, codeword length $n$, average probability of error $\epsilon$ and codewords belonging to $\mathcal{F}^{(n)}$,  for the same block fading channel under consideration with CSIR but \textit{no} CSIT. These codewords satisfy the $\mathbf{ST}$ constraint. We use $\mathcal{C}(M,n,\epsilon,1)$ in conjunction with a power controller. Next, we explain the power control scheme. 
\par Choose $\lambda_B$ such that $\mathbb{E}_{H}\big[\big(\lambda_B-\frac{\sigma_N^2}{H^2}\big)^+\big]=\bar{P}-\delta_B,$ for some {$\delta_B$} (to be chosen appropriately later {with $\bar{P}>\delta_B>0$}). Let ${\mathcal{P}}_{\text{WF}}^{(\text{B})}({H})=\big( \lambda_B-\frac{\sigma_N^2}{H^2}\big)^+$. At the beginning of block $b$, the transmitter checks if the constraint $$\mathbf{C_{ST}}~:~\sum\limits_{l=1}^{b}\frac{||\underline{x}_l'||^2}{n_c}\Big(\lambda_B-\frac{\sigma_N^2}{H_l^2}\Big)^+ \leq B\bar{P}$$ is met. If $\mathbf{C_{ST}}$ is not met, the $\mathbf{ST}$ constraint  is violated. The transmission is halted and an \textit{error is declared} (a short error
message is sent to the receiver which it can receive without
error). Else, for the $k^{\text{th}}$ transmission in block $b$, the channel input $X_{[bk]}$ is chosen such that $X_{[bk]}=x_{[bk]}'\sqrt{\mathcal{P}_{\text{WF}}^{(\text{B})}(H_b)}$. If the transmitter sends codeword symbols successfully in all $B$ blocks without violating $\mathbf{C_{ST}}$ for all $b \in [1:B]$, $$\frac{1}{Bn_c}\sum\limits_{b=1}^{B}\sum\limits_{k=1}^{n_c}x_{[bk]}^2=\frac{1}{B}\sum\limits_{b=1}^{B}\frac{||\underline{x}_b'||^2}{n_c}\mathcal{P}_{\text{WF}}^{(\text{B})}(H_b)\leq \bar{P}$$ and the $\mathbf{ST}$ constraint is met. If transmission error does not occur, for all $b \in [1:B],~k\in[1:n_c]$, the channel output $Y_{[bk]}=H_{b}\sqrt{\mathcal{P}_{\text{WF}}^{(\text{B})}(H_b)}x_{[bk]}'+Z_{[bk]}.$ 
\par Let $G'\stackrel{D}{=}G_b' \triangleq H_b\sqrt{\mathcal{P}_{\text{WF}}^{(\text{B})}(H_b)}$. Any rate achievable over a block fading channel with i.i.d. fading process $\{G_b'\}$ with CSIR and \textit{no} CSIT satisfying an $\mathbf{ST}$ constraint of unity, is also achievable, over a channel subject to block fading process $\{H_b\}$ with CSIT and CSIR satisfying an $\mathbf{ST}$ constraint of $\bar{P}$. Henceforth, we consider the channel with block fading process $\{G_b'\}$ (having CSIR but no CSIT).
\par   Let $\delta_B'$ denote an upper bound on the probability of the event that $\mathbf{C_{ST}}$ is violated. A specific choice of $\delta_B'$ will be made later. Let $\tau'=\tau+\delta_B'$, for some $\tau>0$. It is immediate to see that the following bound (Theorem 25, \cite{polyanskiy2010channel}) holds: for $\epsilon'\triangleq \epsilon-\delta_B'$, $0<\tau<\epsilon'$, 
 \begin{equation}
\label{kay_beta}
M \geq \frac{\kappa_{\tau}\Big(\mathcal{F}^{(n)},Q_{\mathbf{Y},\mathbf{G}}\Big)}{\sup_{\mathbf{x}'\in\mathcal{F}^{(n)}}\beta_{1-\epsilon'+\tau}\Big(\mathbf{x}',Q_{\mathbf{Y},\mathbf{G}}\Big)},
\end{equation}
 where {$\beta_{\alpha}(\mathbf{x}',Q) \equiv \displaystyle\beta_{\alpha}(P_{\mathbf{Y}, \mathbf{G} \vert \mathbf{X}=\mathbf{x}'}(\mathbf{y},\mathbf{g}\vert \mathbf{x}'),Q) $,  $Q_{\mathbf{Y},\mathbf{G}}$ is an auxiliary output distribution, $\kappa_{\tau}(\mathcal{F}^{(n)},Q)$ are as defined in \cite{polyanskiy2010channel} and $ \displaystyle P_{\mathbf{Y},\mathbf{G} \vert \mathbf{X}=\mathbf{x}'}(\mathbf{y},\mathbf{g}\vert \mathbf{x}')$ corresponds to the channel transition probability.} Now, from the analysis in \cite{polyanskiy2011scalar}, it follows that 
\begin{equation}
\label{ST_M_bnd}
\log M^* \geq n\mathbf{C}\big(\bar{P}-\delta_B\big)+\sqrt{nV_{\text{BF}}^t}\Phi^{-1}(\epsilon'-\tau)+o(\sqrt{n}),
\end{equation}
where $V_{\text{BF}}^t\equiv V_{\text{BF}}^t(\bar{P}-\delta_B)=\mathbb{E}_{\text{G}'}\big[V(G'^2)\big]+n_c\mathbb{V} \big[C(G'^2)\big]+\frac{1}{2}\mathbb{V}\big[\mathcal{L}(G'^2)\big]$. Note that the difference between $V_{\text{BF}}^t$ and $V_{\text{BF}}(\bar{P})$ (defined in the statement of Theorem \ref{Th_ST_LB2}) is that, $V_{\text{BF}}^t$ is evaluated according to the distribution of $G'$ rather than $G$ defined in Theorem \ref{Th_ST_LB2}. {As per definition, $G'$ is evaluated such that $\mathbb{E}_{\text{H}}[\mathcal{P}_{\text{WF}}^{(\text{B})}(H)]=\bar{P}-\delta_B$.}  
 The \textit{first order term} in the lower bound obtained does not match with the channel capacity and hence we proceed to rectify this shortcoming. For convenience, let $\delta_B'' =\delta_B/\bar{P}.$ {Here, $\delta_B'' \in(0,1)$.} Note that, for $G$ as defined in Theorem \ref{Th_ST_LB2},
 \begin{flalign*}
 \mathbf{C}(\bar{P}-\delta_B)&=\mathbb{E}_{\text{G}'}\Big[C\big({G'}^2\big)\Big]\stackrel{(a)}{=}\mathbb{E}_{\text{G}}\Big[C\big(G(1-\delta_B'')\big)\Big].
 \end{flalign*}
 Here $(a)$ follows because of the following reason: $$\mathbb{E}_{\text{H}}\Big[\mathcal{P}_{\text{WF}}(H)\Big]=\bar{P} \Rightarrow \mathbb{E}_{\text{H}}\Big[\mathcal{P}_{\text{WF}}(H)\big(1-\delta_B''\big)\Big]=\bar{P}-\delta_B.$$ {For $\epsilon''>0$ and function $\eta$, Taylor's theorem applied to the \textit{functional} $F(.)$ about a certain function $f$ is given by $$F(f+\epsilon''\eta)=F(f)+\epsilon''\frac{dF(f+y\eta)}{dy}\Bigg|_{y=x},$$ {for some $x \in(0,\epsilon'')$.} Applying Taylor's theorem to the {functional} $\mathbf{C}(\bar{P}-\delta_B)$ yields} $${\mathbf{C}(\bar{P}-\delta_B)}{{=}\mathbf{C}(\bar{P})-\delta_B''\mathbb{E}_{\text{G}}\Big[{\mathcal{L}(G(1-\alpha_1\delta_B''))}\big/{2(1-\alpha_1\delta_B'')}\Big],}$$
for $\alpha_1 \in (0,1)$. Next, we apply Taylor's theorem to the functional $\sqrt{V_{\text{BF}}^t(\bar{P}-\delta_B)}${, and obtain } $$\sqrt{V_{\text{BF}}^t(\bar{P}-\delta_B)}{=}\sqrt{V_{\text{BF}}(\bar{P})}-\delta_B''F_1(\bar{P},\alpha_2,\delta_B)$$
for some $\alpha_2 \in (0,1)$. Here, $$F_1(\bar{P},\alpha_2,\delta_B)=F_2(\bar{P},\alpha_2,\delta_B)\big/{2\sqrt{V_{\text{BF}}(\bar{P}-\alpha_2\delta_B)}},$$ where, $V_{\text{BF}}(\bar{P}-\alpha_2\delta_B)$ is evaluated in the same way as $V_{\text{BF}}^t$ and $V_{\text{BF}}$, but such that $\mathbb{E}_{\text{H}}\big[\big(\lambda'-\frac{\sigma_N^2}{H^2}\big)^+\big]=\bar{P}-\alpha_2\delta_B$. Also, with $\bar{\delta}_B\triangleq(1-\alpha_2{\delta}_B'')$, $\bar{G}\triangleq G^2\bar{\delta}_B$, $\tilde{G}\triangleq (\sigma_N^2+\bar{G})$  and $L_1 \triangleq \mathbb{E}_{\text{G}}\big[{\sigma_N^4G^2}\big/{\tilde{G}^3}\big]$, $F_2(\bar{P},\alpha_2,\delta_B)\equiv F_2$ is given by
\begin{flalign*} 
L_1+n_c\text{Co}\big(C(\bar{G}),{\mathcal{L}(\bar{G})}\big/{\bar{\delta}_B}\big)+\text{Co}\big(1-\mathcal{L}(\bar{G}),{\sigma_N^2\mathcal{L}\big(\bar{G})}\big/{\tilde{G}^2}\big).
\end{flalign*} 
Here $\text{Co}(.,.)$ denotes the covariance function.
 Finally, we apply Taylor's theorem to the function $\Phi^{-1}(x),~x \in[\epsilon',\epsilon]$ so that $\Phi^{-1}(\epsilon')=\Phi^{-1}(\epsilon)-\delta_B'g_1(\epsilon,\delta_B'),$ for some $\alpha_3 \in (0,1)$ and $g_1(\epsilon,\delta_B')=\sqrt{2\pi}\text{exp}\big\{{\Phi^{-1}({\epsilon-\alpha_3\delta_B'})}\big/{\sqrt{2}}\big\}^2.$ For notational convenience, let $g_2\equiv g_2(\epsilon,\delta_B') \triangleq \delta_B'g_1(\epsilon,\delta_B')$, $F_3\equiv F_3(\bar{P},\alpha_2,\delta_B)\triangleq \delta_B''F_1(\bar{P},\alpha_2,\delta_B)$, $\hat{\delta}_B \triangleq (1-\alpha_1\delta_B'')$, $\hat{G} \triangleq \hat{\delta}_BG^2$ and $\mathcal{L}_1(\hat{\text{G}})\triangleq {\mathcal{L}(\hat{\text{G}})}{(2\hat{\delta}_B)}^{-1}$. Substituting these values in (\ref{ST_M_bnd}), we obtain 
\begin{flalign*}
\hspace{-8pt}
\log M^*&\geq n\mathbf{C}(\bar{P})-n\delta_B''\mathbb{E}_{\hat{\text{G}}}\Big[\mathcal{L}_1(\hat{\text{G}})
\Big]+\sqrt{nV_{\text{BF}}(\bar{P})}\Phi^{-1}(\epsilon)\\
&+\sqrt{n}F_3\Phi^{-1}(\epsilon)-\sqrt{n{V_{\text{BF}}}(\bar{P})}g_2{-\sqrt{n}g_2F_2+o(\sqrt{n}).}
\end{flalign*}
\par Next, we fix the choice of $\delta_B$ so as to get an exact expression for $\delta_B'$ and $\delta_B''$. To that end, we upper bound the probability of violating the constraint $\mathbf{C_{ST}}$ in the following way: 
 \begin{flalign*}
& \mathbb{P}\Bigg(\bigcup_{b=1}^{B}\Bigg\{\sum\limits_{l=1}^{b}\frac{||\underline{x}'_b||^2}{n_c}\Big({\lambda_B}-\frac{\sigma_N^2}{H_l^2}\Big)^+ > B\bar{P}\Bigg\}\Bigg)\\
 &=\mathbb{P}\Bigg(\sum\limits_{b=1}^{B}\frac{||\underline{x}'_b||^2}{n_c}\mathcal{P}_{\text{WF}}^{(\text{B})}(H_b) > B\bar{P}\Bigg)\\
 &\leq \mathbb{P}\Bigg(\sum\limits_{b=1}^{B}\Big\{\frac{||\underline{x}'_b||^2}{n_c}\mathcal{P}_{\text{WF}}^{(\text{B})}(H_b)-(\bar{P}-\delta_B)\Big\} > B\delta_B\Bigg)\\
 &=\mathbb{P}\Bigg(\sum\limits_{b=1}^{B}S_b > B\delta_B\Bigg)\stackrel{(a)}{\leq} \text{exp}{\Big(-\frac{B\delta_B^2}{2{\lambda}^2}\Big)}\stackrel{(b)}{\leq}\frac{\kappa}{B}, 
 \end{flalign*}      
 where, $S_b \triangleq \frac{||\underline{x}'_b||^2}{n_c}\mathcal{P}_{\text{WF}}^{(\text{B})}(H_b)-(\bar{P}-\delta_B) $ and $(a)$ follows from Hoeffding's inequality (\cite{boucheron2013concentration}) as $\{S_b\}$ is a set of i.i.d. random variables with zero mean such that $\sum\limits_{b=1}^{B}S_b \in [-\lambda_B,{\lambda_B}] \subseteq [-\lambda,{\lambda}]$. For some constant $\kappa$, $(b)$ follows by choosing $\delta_B=\lambda\sqrt{{2}/{{B^{(1-\alpha)}}}}$, for an arbitrary, small $\alpha>0$.  For the specified choice of $\delta_B$, by Taylor's theorem, 
 \begin{flalign*}
 \mathbb{E}_{{\hat{G}}}\Big[{\mathcal{L}(\hat{G})}{(2\hat{\delta}_B)}^{-1}\Big]=\mathbb{E}_{\text{G}}\Big[{\mathcal{L}(G)}/2\Big]-\delta_B''\mathbb{E}_{\text{G}}\Big[{\mathcal{L}_1'\big(G(1-u)\big)}\Big],
 \end{flalign*} 
 for some $u \in(0,\alpha_1\delta_B'')$ and $\mathcal{L}_1'\big(G(1-u)\big)$ denotes the first derivative of $\mathcal{L}_1\big(G(1-x)\big)$ with respect to $x$ at $x=u$. Gathering all the $o(\sqrt{n})$ terms, we obtain that $\log  M^*$ is upper bounded by $n\mathbf{C}(\bar{P})-\sqrt{{n^{(1+\alpha)}}/{2}}+\sqrt{n{V_\text{BF}}(\bar{P})}\Phi^{-1}(\epsilon)+o(\sqrt{n}).$ Since $\alpha$ is chosen arbitrarily small, the result follows.
\qed

 \section*{Appendix B}
\subsection*{Lower bound: $\mathbf{LT}$ constraint}
\label{App_Ach}
\par As in Appendix A, we fix the constraint set $\mathcal{F}^{(n)}.$ Let $\mathcal{C}({M},n,\epsilon,1)$ be a codebook with $M$ codewords, codeword length $n$, average probability of error $\epsilon$ and codewords belonging to $\mathcal{F}^{(n)}$,  for a block fading channel (same as the one in consideration) with CSIR but \textit{no} CSIT. By virtue of belonging to $\mathcal{F}^{(n)}$, the codewords inherently satisfy the $\mathbf{ST}$ constraint. We use this codebook in conjunction with a \textit{power controller} in the following way: at the beginning of block $b$, the channel gain $H_b$ for the next $n_c$ channel uses will be available to the transmitter. Fix the water filling power allocation $\mathcal{P}_{\text{WF}}(.)$ in equation (\ref{gold_eqn}). For the $k^{\text{th}}$ transmission in block $b$, the channel input $X_{[bk]}$ is chosen such that $X_{[bk]}=x_{[bk]}'\sqrt{\mathcal{P}_{\text{WF}}(H_b)}$. Note that $$\mathbb{E}_{\mathbf{H}}\Big[\sum\limits_{b=1}^{B}\sum\limits_{k=1}^{n_c}\mathcal{P}_{\text{WF}}(H_b){{x}_{[bk]}'}^2\Big]=\mathbb{E}_{\text{H}}\big[\mathcal{P}_{\text{WF}}(H)\big]||\mathbf{x}'||^2 =  n\bar{P}.$$ For each joint fading state $\mathbf{H} \in \mathbb{R}^B$, the above scheme generates a codebook satisfying $\mathbf{LT}$ constraint for the channel with CSIT and CSIR \textit{on the fly}, in a causal manner. The channel output $$Y_{[bk]}=H_{b}\sqrt{\mathcal{P}_{\text{WF}}(H_b)}x_{[bk]}'+Z_{[bk]}.$$ 
 \par Now, consider a block fading channel with i.i.d. fading process $\Big\{H_{b}\sqrt{\mathcal{P}_{\text{WF}}(H_b)}\Big\}$. Assume no CSIT and full CSIR. Any rate achievable over such a channel with $\mathbf{ST}$ constraint  is achievable over the original channel (with full CSIT and CSIR) subject to the $\mathbf{LT}$ constraint as well.  It follows from the analysis in \cite{polyanskiy2011scalar} that $$\log M^*\geq  n\mathbf{C}(\bar{P})+\sqrt{nV_{\text{BF}}(\bar{P})}\Phi^{-1}(\epsilon)+ o(\sqrt{n}).$$ \qed
 
\section*{Appendix C}
\subsection*{Upper bound: $\mathbf{ST}$ constraint}
Let $M^* \equiv M^*(n,\epsilon,\bar{P})$ denote the maximal codebook size for the channel with average probability of error $\epsilon and$ satisfying the $\mathbf{ST}$ constraint. In particular, assume that the $\mathbf{ST}$ constraint is satisfied with equality. Note that we can make this assumption without loss of optimality (see, for instance, \cite{polyanskiy2010channel1}, Lemma 65). Next, assume that CSIT is known non-causally. Bound derived under this assumption will give a valid upper bound for the causal case under consideration. Corresponding to each codeword $\mathbf{x}$ (satisfying $\mathbf{ST}$ constraint) and $\mathbf{h}$ (obtained non causally), we can identify a power allotment vector $\mathcal{P}(\mathbf{x},\mathbf{h})\triangleq \big(\mathcal{P}_1(\mathbf{x},\mathbf{h}),\hdots,\mathcal{P}_{|\mathcal{H}|}(\mathbf{x},\mathbf{h})\big)$, where,
\begin{flalign}
\label{pow_alloc_defn}
  \mathcal{P}_i(\mathbf{x},\mathbf{h})= \left\{
     \begin{array}{lr}
      0 & ~ ,~ n_i(\mathbf{h}) = 0,\\
      \frac{1}{n_i(\mathbf{h})}\sum\limits_{j=1}^{n_i(\mathbf{h})}{x_{k_j(i,\mathbf{h})}^2}        &~ , ~ n_i(\mathbf{h}) > 0.
     \end{array}
   \right.
\end{flalign}
Here, $n_i(\mathbf{h})=n_c\sum\limits_{b=1}^{B}\mathbf{1}_{\{h_b=\eta_i\}}$ and $k_{j}(i,\mathbf{h})=\{k\in[1:n]:h_{k}=\eta_i\} $. Also, $\mathcal{P}(\mathbf{x},\mathbf{h})$ is such that $\sum\limits_{i=1}^{|\mathcal{H}|}t_i(\mathbf{h})\mathcal{P}_i(\mathbf{x},\mathbf{h})=\bar{P},$ where $t_i(\mathbf{h})\equiv t_i={n_i(\mathbf{h})}/{n}$.  Since we assume perfect CSIT and CSIR, let $\bar{\mathbf{X}}=(\mathbf{X},\mathbf{H})$ and $\bar{\mathbf{Y}}=(\mathbf{Y},\mathbf{H})$ denote the equivalent channel input and output, respectively. Then, we have (\cite{polyanskiy2010channel}, Theorem 26) that $$\beta_{1-\epsilon}(P_{\mathbf{X},\mathbf{H},\mathbf{Y}},Q_{\mathbf{X},\mathbf{H},\mathbf{Y}})\leq 1-\epsilon',$$ where {$\beta_{\alpha}(P,Q)$ is the minimum \textit{false alarm probability} in deciding between $P$ and $Q$, subject to a minimum \textit{detection probability} $\alpha>0$}, $Q$ denotes {an} auxiliary channel and $\epsilon'$ is the average probability of error  for the auxiliary channel. {We choose the auxiliary channel where the channel output $\mathbf{Y}$ has the distribution} $$Q_{\mathbf{X},\mathbf{H},\mathbf{Y}}=\prod_{b=1}^{B}\prod_{k=1}^{n_c}\mathcal{N}\Big(0,\sigma_N^2+H_b^2\mathcal{P}_{\text{WF}}(H_b)\Big).$$ Here, $\mathbf{E}_{\text{H}}[\mathcal{P}_{\text{WF}}(H)]=\bar{P}.$ Since the message $M\in\mathcal{M}$ is independent of $\mathbf{H}$, the output of the auxiliary channel is independent of $M$ and hence 
\begin{equation*}
\hspace{20pt}
\beta_{1-\epsilon}(P_{\mathbf{X},\mathbf{H},\mathbf{Y}},Q_{\mathbf{X},\mathbf{H},\mathbf{Y}})\leq 1-\Big(1-\frac{1}{M^*}\Big).
\end{equation*} 
 Next, $\beta_{1-\epsilon}$ can be lower bounded (as in \cite{polyanskiy2010channel}, equation (102)) as,  $\beta_{1-\epsilon}\geq \frac{1}{\gamma}\Big(\mathbb{P}\big[ \mathcal{I}_\gamma\big]-\epsilon\Big)$, where $\gamma>0$, $\mathcal{I}_\gamma \triangleq \Big\{i(\mathbf{X},\mathbf{H},\mathbf{Y})\leq \log \gamma\Big\}$ and $ i(\mathbf{X},\mathbf{H},\mathbf{Y}) \triangleq \log \frac{P_{\mathbf{X},\mathbf{H}, \mathbf{Y}}}{Q_{\mathbf{X},\mathbf{H},\mathbf{Y}}}.$
 \par {Optimizing over the choice of input distributions, for the auxiliary channel under consideration, we have
 \begin{equation}
 \label{conv_ST_ineq}
 \log M^* \leq \log \gamma -\inf\limits_{F_{\mathbf{X}}}\log\Big(P\big[\mathcal{I}_\gamma \big]-\epsilon\Big)^+.
\end{equation}  
 With the prescribed choice of the auxiliary channel, for $\mathbf{H}=\mathbf{h}$, $\mathbf{X}=\mathbf{x}$, under $P_{\mathbf{X},\mathbf{H},\mathbf{Y}}$, $ i(\mathbf{x},\mathbf{h},\mathbf{Y}) \stackrel{D}{=} \sum\limits_{b=1}^{B}W_b',$ where,
 \begin{flalign*}
 \hspace{-10pt}
W_b'={n_c}C(g_b^2)+\frac{\sigma_N^2h_b^2||\underline{x}_b||^2+2h_b\sigma_N^2\langle\underline{x}_b,\underline{Z}_b\rangle-g_b^2||\underline{Z}_b||^2}{2\sigma_N^2\big(\sigma_N^2+g_b^2\big)},
 \end{flalign*}
and $g_b=h_b\sqrt{\mathcal{P}_{\text{WF}}(h_b)}$.
Next, it can be seen that  $W_b'\stackrel{D}{=}W_b''(\mathcal{P}_{h_b}(\mathbf{x},\mathbf{h}))\equiv W_b''$, where, for $p\geq 0$,
\begin{flalign*}
 \hspace{-10pt}
W_b''(p)&\stackrel{D}{=}{n_c}C(g_b^2)+\frac{n_c\sigma_N^2h_b^2p+2h_b\sqrt{p}\sigma_N^2\langle\underline{1}_b,\underline{Z}_b\rangle-g_b^2||\underline{Z}_b||^2}{2\sigma_N^2\big(\sigma_N^2+g_b^2\big)},
 \end{flalign*}
 where, $\mathcal{P}_{h_b}(\mathbf{x},\mathbf{h})$ is as defined in (\ref{pow_alloc_defn}) and $\underline{1}_b$ is the $n_c$ length vector $(1,1,\hdots,1)$ for all $b$.
 Note that $\{W_b''\}$ are independent random variables. Let $F_{\mathbf{h}}\equiv F_{\mathbf{X}\vert\mathbf{H}=\mathbf{h}}$ denote the conditional distribution of $\mathbf{X}$ given $\mathbf{H}=\mathbf{h}.$ Then, $\mathbb{P}\big[ \mathcal{I}_\gamma \big]$
 \begin{flalign}
 \label{prob_LB_seq1}
  &=\sum\limits_{\mathbf{h}\in\mathcal{H}^B}\mathbb{P}(\mathbf{h})\int_{\mathbb{R}^n } \mathbb{P}\Big[ i(\mathbf{x},\mathbf{h},\mathbf{Y})\leq \log \gamma\Big\vert\mathbf{x},\mathbf{h} \Big]dF_{{\mathbf{h}}}(\mathbf{x})\nonumber\\
&=\sum\limits_{\mathbf{h}\in\mathcal{H}^B}\mathbb{P}(\mathbf{h})\int_{\mathbb{R}^n } \mathbb{P}\Bigg[ \sum_{b=1}^{B}W_b' \leq \log \gamma\Bigg\vert\mathbf{x},\mathbf{h} \Bigg]dF_{\mathbf{h}}(\mathbf{x})\nonumber\\
&\stackrel{(a)}{\geq}\sum\limits_{\mathbf{h}\in\mathcal{H}^B}\mathbb{P}(\mathbf{h})\int_{\mathbb{R}^n }\inf\limits_{\mathbf{x}\in\mathcal{S}_{\mathbf{h}}} \mathbb{P}\Bigg[ \sum_{b=1}^{B}W_b' \leq \log \gamma\Bigg\vert\mathbf{x},\mathbf{h} \Bigg]dF_{\mathbf{h}}(\mathbf{x}),
\end{flalign}
where, in $(a)$, $\mathcal{S}_\mathbf{h}$ denotes the support of $F_{\mathbf{h}}$.
Next, note that, 
\begin{flalign}
\label{prob_LB_seq2}
&\inf\limits_{\mathbf{x}\in\mathcal{S}_{\mathbf{h}}} \mathbb{P}\Bigg[ \sum_{b=1}^{B}W_b' \leq \log \gamma\Bigg\vert\mathbf{x},\mathbf{h} \Bigg] \nonumber \\
&\stackrel{(a)}{=} \inf\limits_{\mathcal{P}(\mathbf{h})\in\Pi_{\mathbf{h}}} \mathbb{P}\Bigg[ \sum_{b=1}^{B}W_b''\big(\mathcal{P}_{h_b}(\mathbf{h})\big)\leq \log \gamma\Big\vert \mathbf{h} \Bigg]\nonumber\\
 &\stackrel{(b)}{=}\inf\limits_{\mathcal{P}(\mathbf{h})\in\Pi_{\mathbf{h}}} \mathbb{P}\Bigg[ \frac{1}{{\sqrt{\nu_B}}}\sum_{b=1}^{B}W_b\leq \gamma_B'\big(\mathbf{h},\mathcal{P}(\mathbf{h})\big)\Bigg\vert\mathbf{h} \Bigg]\nonumber \\
 &\stackrel{(c)}{\geq} \inf\limits_{\mathcal{P}(\mathbf{h})\in\Pi_{\mathbf{h}}} \Big[\Phi\Big(\gamma_B'\big(\mathbf{h},\mathcal{P}(\mathbf{h})\big)\Big)-\mathcal{B}\big(\mathbf{h},\mathcal{P}(\mathbf{h})\big)\Big]\nonumber \\
 &\stackrel{(d)}{\geq} \inf\limits_{\mathcal{P}(\mathbf{h})\in\Pi_{\mathbf{h}}} \Phi\Big(\gamma_B'\big(\mathbf{h},\mathcal{P}(\mathbf{h})\big)\Big)- \sup\limits_{\mathcal{P}(\mathbf{h})\in\Pi_{\mathbf{h}}}\mathcal{B}_1\big(\mathbf{h},\mathcal{P}(\mathbf{h})\big)\nonumber \\
 &\stackrel{(e)}{= }\min\limits_{\mathcal{P}(\mathbf{h})\in\Pi_{\mathbf{h}}} \Phi\Big(\gamma_B'\big(\mathbf{h},\mathcal{P}(\mathbf{h})\big)\Big)- \max\limits_{\mathcal{P}(\mathbf{h})\in\Pi_{\mathbf{h}}}\mathcal{B}_1\big(\mathbf{h},\mathcal{P}(\mathbf{h})\big)\nonumber \\
  &\stackrel{(f)}{= }\Phi\Big(\gamma_B'\big(\mathbf{h},\mathcal{P}_1^*(\mathbf{h})\big)\Big)- \mathcal{B}_1\big(\mathbf{h},\mathcal{P}_2^*(\mathbf{h})\big)
 \end{flalign}
 where, $(a)$ follows from the fact that $W_b'\stackrel{D}{=}W_b''$ and the definition of $W_b''$ and $\Pi_\mathbf{h}$, where, 
 $$\Pi_\mathbf{h}\triangleq \Big\{\mathcal{P}(\mathbf{h})\in \mathbb{R}^{|\mathcal{H}|}:\sum\limits_{i=1}^{|\mathcal{H}|}t_i(\mathbf{h})\mathcal{P}_i({\mathbf{h}})=\bar{P}\Big\}.$$ In $(b)$, $W_b=(W_b''\big(\mathcal{P}_{h_b}(\mathbf{h})\big)-\mu_b)$, $\gamma_B'\big(\mathbf{h},\mathcal{P}(\mathbf{h})\big)\equiv \gamma_B'(\mathbf{h})=({\log \gamma-\mu(\mathbf{h})})/{\sqrt{\nu(\mathbf{h})}}$, the mean of $S_B({W'})$ (as noted, $S_B(.)$ is our alternate notation for summation), $\mu(\mathbf{h})=S_B({\mu_B})$, where, $\mu_b \equiv \mu_{h_b}(\mathbf{h}) ={n_c}C(g_b^2)+{(h_b^2n_c\mathcal{P}_{h_b}(\mathbf{h})-g_b^2n_c)}\big/{2\big(\sigma_N^2+g_b^2\big)}.$ Similarly,  the variance of $S_B({W'})$, $\nu(\mathbf{h})=S_B({\nu})$, where, 
 \begin{equation}
 \label{variance_rel}
 \nu_b\equiv \nu_{h_b}(\mathbf{h})=\big({2\sigma_N^2h_b^2n_c\mathcal{P}_{h_b}(\mathbf{h})+n_cg_b^4}\big)\big/{2\big(\sigma_N^2+g_b^2\big)^2}.
 \end{equation}  Next, $(c)$ follows from Berry Esseen Theorem for independent, but not identically distributed random variables (\cite{feller}, Chapter 16) with $\mathcal{B}(\mathbf{h})\equiv\mathcal{B}\big(\mathbf{h},\mathcal{P}(\mathbf{h})\big)$, where,
\begin{equation}
\label{Berry}
\mathcal{B}(\mathbf{h})={C\tau(\mathbf{h})}\big/{\big(\nu(\mathbf{h})\big)^{3/2}},~\tau(\mathbf{h})\triangleq \sum\limits_{b=1}^{B}\mathbb{E}_{\underline{Z}_b}\big[|W_b|^3\big],
\end{equation} 
  for some constant $C>0$ (not depending on $\mathbf{h})$. The term $\mathcal{B}_1\big(\mathbf{h},\mathcal{P}(\mathbf{h})\big)$ in (d) is defined in the proof of Lemma \ref{lem_approx_1} provided in Appendix E. Also, from the definition therein, $\mathcal{B}_1\big(\mathbf{h},\mathcal{P}(\mathbf{h})\big)\geq \mathcal{B}\big(\mathbf{h},\mathcal{P}(\mathbf{h})\big)$. Together with the fact that, taking infimum of the first term and supremum of the second term (on the right hand side of  $(c)$) separately can only lower the value, $(d)$ follows.  In $(e)$, the infimum and supremum are replaced with minimum and maximum respectively, by noting the following. The functions $\mathcal{B}_1$ (from the definition in Lemma \ref{lem_approx_1} in Appendix E )) and $\Phi$, are continuous in $\mathcal{P}(\mathbf{h})$. Also, $\Pi_\mathbf{h}\subset\mathbb{R}^{|\mathcal{H}|}$ is a compact set. Hence the minimum and the maximum respectively, are attained. Also, note that the minimum in $(e)$ is not 0 due to the equality constraint in the definition of $\Pi_\mathbf{h}$.  Finally, in $(f)$, $\mathcal{P}_1^*(.)$ and $\mathcal{P}_2^*(.)$ denote the power allocations that attain the minimum and maximum, respectively. 
\par From (\ref{prob_LB_seq1}) and (\ref{prob_LB_seq2}), we have that
\begin{flalign}
\label{Prob_I_gamma_LB}
\mathbb{P}\big[I_\gamma\big]&\geq \sum_{\mathbf{h}\in \mathcal{H}^B}\mathbb{P}(\mathbf{h})\Big[\Phi\Big(\gamma_B'\big(\mathbf{h},\mathcal{P}_1^*(\mathbf{h})\big)\Big)- \mathcal{B}_1\big(\mathbf{h},\mathcal{P}_2^*(\mathbf{h})\big)\Big]\nonumber \\
&\stackrel{(a)}{\geq}\mathbb{E}_{\mathbf{H}}\Big[\Phi\Big(\gamma_B'\big(\mathbf{H},\mathcal{P}_1^*(\mathbf{H})\big)\Big)\Big]-\frac{C_2}{\sqrt{n}},
\end{flalign} 
where, $(a)$ follows from  Lemma \ref{lem_approx_1} stated and proved in Appendix E. Here, $C_2$ is a universal positive constant. 
}
 \par Next, we proceed to obtain a tractable lower bound for  $\gamma_B'\big(\mathbf{H},\mathcal{P}_1^*(\mathbf{H})\big)$. To that end, we make use of the following two facts. First, since $\mathcal{P}_{\text{WF}}(h_b)\geq (\lambda-\frac{\sigma_N^2}{h_b^2}\Big)$,  
\begin{equation}
\label{rel_WF_ineq}
\big(\sigma_N^2+g_b^2\big)\geq \Big(\sigma_N^2+h_b^2\big(\lambda-{\sigma_N^2}\big/{h_b^2}\big)\Big)= \lambda h_b^2.
\end{equation}
 Also, using the fact that $\mathcal{P}_{\text{WF}}(h_b)=\mathbf{1}_{\{h_b^2\geq \lambda\}}\mathcal{P}_{\text{WF}}(h_b)$, it can be seen that
\begin{flalign}
\label{rel_waterfill}
\mathcal{L}(g_b^2)=\frac{\mathcal{P}_{\text{WF}}(h_b)}{\lambda}.
\end{flalign}
Using the above two relationships, it can be seen that
\begin{flalign*}
{n_c}C(g_b^2)+\frac{h_b^2n_c\mathcal{P}_{1,h_b}^*(\mathbf{h})-g_b^2n_c}{2\big(\sigma_N^2+g_b^2\big)}\leq \mu_b'(h_b),
\end{flalign*}
where, $\mu_b'(h_b)\equiv \mu_b'={n_c}C(g_b^2)+{\big(n_c\bar{P}-n_c\mathcal{P}_{\text{WF}}(h_b)\big)}/{2\lambda}$. Define $\sum\limits_{b=1}^{B}\mu_b'(h_b)\triangleq \mu'(\mathbf{h})$. Let  $R \triangleq \log \gamma /n$, $\mu''(\mathbf{h})=\mu'(\mathbf{h})/n$, $\nu'\big(\mathbf{h},\mathcal{P}_1^*(\mathbf{h})\big)=\nu\big(\mathbf{h},\mathcal{P}_1^*(\mathbf{h})\big)/n$. Then, 
\begin{equation}
\label{Psi_def}
\Phi\Big(\gamma_B'\big(\mathbf{h},\mathcal{P}_1^*(\mathbf{h})\big)\Big)\geq \Phi\Bigg(\sqrt{n}\frac{R-\mu''(\mathbf{h})}{\sqrt{\nu'\big(\mathbf{h},\mathcal{P}_1^*(\mathbf{h})\big)}}\Bigg)\triangleq \Psi(\mathbf{h}).
\end{equation} 
Here, the inequality follows from the fact that $\Phi(.)$ is a monotonically increasing function. 
\par Now, using Lemma \ref{lem_Berry} provided in Appendix F, we obtain, $$\mathbb{E}_{\mathbf{H}}\big[{\Psi(\mathbf{H})}\big]\geq \Phi\Bigg(\sqrt{n}\frac{R- \mathbf{C}(\bar{P})}{\sqrt{V_{\text{BF}}'(\bar{P})}}\Bigg)-C_3\frac{\log n}{n},$$ where, $V_{\text{BF}}'(\bar{P})$ is defined as in the statement of Theorem \ref{Th_ST_UB1}. Choose $R=\mathbf{C}(\bar{P})+\sqrt{\frac{V_{\text{BF}}'(\bar{P})}{n}}\Phi^{-1}\Big(\epsilon+C_3\frac{\log n}{n}+\frac{2C2}{\sqrt{n}}\Big)$. With the prescribed choice of $R$, combining (\ref{conv_ST_ineq}) and (\ref{Prob_I_gamma_LB}) and applying Taylor's theorem to the function $\Phi^{-1}(.)$, we obtain, ${\log M}^*\leq n\mathbf{C}(\bar{P})+\sqrt{nV_{\text{BF}}'(\bar{P})}\Phi^{-1}(\epsilon)+O({\log n}).$ \qed
\section*{Appendix D}
\subsection*{Upper bound: $\mathbf{LT}$ constraint }
Let $M^* \equiv M^*(n,\epsilon,\bar{P})$ denote the maximal codebook size for the channel with average probability of error $\epsilon$ satisfying the $\mathbf{LT}$ constraint. In particular, as in the $\mathbf{LT}$ case, assume that the $\mathbf{LT}$ constraint is satisfied with equality. As mentioned, this assumption can be made without loss of optimality (see, for instance, \cite{polyanskiy2010channel1}, Lemma 65). Next, assume that CSIT is known non-causally. Bound derived under this assumption will give a valid upper bound for the causal case under consideration.
\par While $\mathbf{ST}$ constraints insists that $||\mathbf{x}||^2=n\bar{P}$ for \textit{every} fading realization $\mathbf{h}$ and channel input vector $\mathbf{x}$, the $\mathbf{LT}$ constraint only requires $||\mathbf{x}||^2=n\mathcal{P}(\mathbf{h})$, for some function $\mathcal{P}:\mathcal{H}^B\rightarrow \mathbb{R}_+$ such that $\mathbb{E}_{\mathbf{H}}[\mathcal{P}(\mathbf{H})]=\bar{P}$. Fix $\delta_n>0$, $\kappa>0$ and $c>\bar{P}+\delta$. For deriving the upper bound, we will first restrict to the class of \textit{policies} ${{\Pi_{\kappa,c}}}$, where, $${\Pi_{\kappa,c}}\triangleq \Big\{\mathcal{P}:\mathcal{H}^B\rightarrow \big[0,\kappa+c\big]: \mathbb{E}_{\mathbf{H}}[\mathcal{P}(\mathbf{H})]=\bar{P}\Big\}.$$ Note that, for any $\mathcal{P}:\mathcal{H}^B\rightarrow \mathbb{R}_+$ such that $\mathbb{E}_{\mathbf{H}}[\mathcal{P}(\mathbf{H})]=\bar{P}$, $\hat{\mathcal{P}}(\mathbf{h})=\min\{\kappa+c,\mathcal{P}(\mathbf{h})\}$ belongs to $\Pi_{\kappa,c}$. As $\kappa\rightarrow \infty$, $\Pi_{\kappa,c}$ converges to the original class of policies on $\mathbb{R}_+$. Any $\mathcal{P}\in \Pi_{\kappa,c}$ is a function of the realizations of an i.i.d. random vector $\mathbf{H}$, such that $0\leq \mathcal{P}(\mathbf{H})\leq \kappa+c$ (where, the inequalities hold in a.s. sense). Hence, in particular, $\mathcal{P}$ satisfies the \textit{bounded difference} condition (\cite{boucheron2013concentration}, Chapter 3). Let $\mathcal{E}_n\triangleq\{\mathbf{h}:\mathcal{P}(\mathbf{h})>\bar{P}+\delta_n\}$. From McDiarmid's inequality (\cite{boucheron2013concentration}),
\begin{equation}
\label{Conc_bound_LT}
\mathbb{P}(\mathcal{E}_n)\leq \delta_n',~ \delta_n'\equiv \delta_n'(\kappa,c)\triangleq \exp \Big({-\frac{2n\delta_n^2}{(\kappa+c)^2}}\Big).
\end{equation}
Note that, under the event $\mathcal{E}_n^c$, the channel input vector $\mathbf{x}$ satisfies the $\mathbf{ST}$ constraint with power $\bar{P}+\delta_n$. Hence, in deriving the upper bound for the $\mathbf{LT}$ case, we can lower bound the probability term $\mathbb{P}\big[I_{\gamma}\big]$ in (\ref{conv_ST_ineq}) by $\mathbb{P}\big[I_\gamma\cap\mathcal{E}_n^c \big]$. It follows from the analysis for the $\mathbf{ST}$ case that $$\mathbb{P}\big[I_\gamma\cap\mathcal{E}_n^c\big]\geq \mathbb{P}(\mathcal{E}_n^c)\mathbb{P}\Bigg(Z\leq \sqrt{n}\frac{R-\mathbf{C}(\bar{P}+\delta_n)}{\sqrt{V_{\text{BF}}'(\bar{P}+\delta_n)}}\Big\vert\mathcal{E}_n^c \Bigg)-C_3\frac{\log n}{n},$$ where, $Z\sim \mathcal{N}(0,1)$.
 Using the fact that, for events $A,~B$ $\mathbb{P}(A\cap B)\geq \mathbb{P}(A)-\mathbb{P}(B^c)$, we obtain $$\mathbb{P}\big[I_\gamma\big]\geq \Phi\Bigg(\sqrt{n}\frac{R-\mathbf{C}(\bar{P}+\delta_n)}{\sqrt{V_{\text{BF}}'(\bar{P}+\delta_n)}}\Bigg)-\delta_n'-C_3\frac{\log n}{n}.$$ Now, we will choose $n$ so that $\epsilon+\delta_n'<\frac{1}{2}$. Noting the fact that $V_{\text{BF}}(\bar{P})$ is monotonically increasing in $\bar{P}$, adopting the same lines of arguments as in the $\mathbf{ST}$ case, we obtain, $$\frac{\log M}{n}^*\leq \mathbf{C}(\bar{P}+\delta_n)+\sqrt{\frac{V_{\text{BF}}'(\bar{P})}{n}}\Phi^{-1}\big(\epsilon+\delta_n'(\kappa,c)\big)+O\Big(\frac{\log n}{n}\Big).$$ Applying Taylor's theorem to $\mathbf{C}(\bar{P}+\delta_n)$, for $\delta_n=n^{\frac{-(1-\alpha)}{2}},$ for some small $\alpha>0$, as in Appendix A, we obtain $\mathbf{C}(\bar{P}+\delta_n)\leq \mathbf{C}(\bar{P})+\sqrt{\frac{1}{4\lambda^2n^{(1-\alpha)}}}$. By virtue of choice of $\delta_n$, $\delta_n'=\exp\big(-2n^\alpha/(\kappa+c)^2\big)$.  Thus, $\frac{\log M^*}{n}$ is upper bounded by  $$ \mathbf{C}(\bar{P})+{\frac{1}{2\lambda n^{\frac{(1-\alpha)}{2}}}}+\sqrt{\frac{V_{\text{BF}}'(\bar{P})}{n}}\Phi^{-1}\big(\epsilon+\delta_n'(\kappa,c)\big)+O\Big(\frac{\log n}{n}\Big).$$ Note that, $\delta_n'(\kappa,c)$ is monotonically increasing in $\kappa$. Since the bound holds good for all $\kappa$, we can take the infimum over $\kappa$ to obtain the tightest bound. Next, from Taylor's theorem, $\Phi^{-1}(\epsilon+\delta_n')\leq \Phi^{-1}(\epsilon)+\delta_n'\big(1/\phi(\Phi^{-1}(\epsilon)\big)$. Finally, performing Taylor's expansion of the function $({1}\big/{2\lambda\sqrt{n^{(1-\alpha)}})}+\sqrt{{V_{\text{BF}}'(\bar{P})}\big/[n{\phi^2(\Phi^{-1}(\epsilon))]}}\exp(-2n^\alpha/c^2)$ around $\alpha=0$ yields the required result.
\section*{Appendix E}
\begin{lem}
\label{lem_approx_1}
Let $\Psi(\mathbf{h})\equiv \Psi\big(\mathbf{h},\nu'\big(\mathbf{h},\mathcal{P}_1^*(\mathbf{h})\big)\big)$ be as defined in (\ref{Psi_def}) and $V_{\text{BF}}'(\bar{P})$ as in the statement of Theorem \ref{Th_ST_UB1}. Then,
$$\mathbb{E}_{\mathbf{H}}\big[{\Psi(\mathbf{H})}\big]\geq \Phi\Bigg(\sqrt{n}\frac{R- \mathbf{C}(\bar{P})}{\sqrt{V_{\text{BF}}'(\bar{P})}}\Bigg)-C_3\frac{\log n}{n},$$
 where, $C_3$ is some positive constant. . 
\end{lem} 
\begin{proof}
The proof follows along similar lines of arguments as in \cite{tomamichel2014second} (Lemma 17 and Lemma 18). However, in contrast to the case therein, the function $\nu'\big(\mathbf{H},\mathcal{P}_1^*(\mathbf{H})\big)$ is not a sum of i.i.d. random variables. This is due to the dependence of the power allocation policy $\mathcal{P}_1^*$ non causally on $\mathbf{H}$. We will show how to circumvent this problem and adapt the proof to our case. To that end, define the \textit{high probability} set $ \mathcal{T}_B\triangleq \Big\{\mathbf{h} \in \mathbb{R}^B: \max\limits_{i\in[1:|\mathcal{H}|]}\big|t_i(\mathbf{h})-q_i\big| \leq \sqrt{\frac{\log B}{B}}\Big\}$. From Hoeffding's inequality (\cite{boucheron2013concentration}), it follows that $\mathbb{P}({\mathcal{T}_B}^c)\leq (2|\mathcal{H}|)/B^2$. 
\par For notational convenience, let $\nu'\big(\mathbf{h},\mathcal{P}_1^*(\mathbf{h})\big)\equiv \nu'$. Then, Taylor's expansion of $\Psi\big(\mathbf{h},\nu'\big)$, as a function of $\nu'\big(\mathbf{h},\mathcal{P}_1^*(\mathbf{h})\big)$ around $\nu_1$ yields $$\Psi\big(\mathbf{H},\nu'\big)=\Psi\big(\mathbf{H},\nu_1\big)+\sum \limits_{k=1}^\infty{\Big(\frac{\nu'-\nu_1}{\nu_1}\Big)}^k\frac{\Psi^{k}(\mathbf{H},\nu_1)}{2^kk!},$$ where, $\Psi^{k}(\mathbf{H},\nu_1)$ is the $k^{\text{th}}$ derivative of $\Psi(\mathbf{H},x)$ at $x=\nu_1$. Let $S_{\infty}(\mathbf{H})$ denote the summation on the right hand side of the above equation, Note that $|S_\infty(\mathbf{H})|\leq 2$. Hence, $\mathbb{E}_{\mathbf{H}}\big[\Psi(\mathbf{H},\nu')\big]$ is upper bounded by $$\mathbb{E}_{\mathbf{H}}\big[\Psi(\mathbf{H},\nu_1)\big]+{\frac{4|\mathcal{H}|}{B^2}}+\mathbb{E}_{\mathbf{H}}\big[S_\infty(\mathbf{H})\big\vert{\mathcal{T}_B}\big].$$ Now, we proceed to obtain a tractable upper bound on $(\nu'-\nu_1)/\nu_1$.  Note that $\mathbb{E}_{\mathbf{H}}[\nu']\neq \nu_1$ and hence we cannot invoke the same lines of arguments as in  \cite{tomamichel2014second}. To that end, observe that, from Taylor's theorem, $V_i^*\leq V_i+V_i'({\mathcal{P}_{1,i}^*-\mathcal{P}_{\text{WF},i}}),$ where, $\mathcal{P}_{1,i}^*\equiv \mathcal{P}_1^*(i,\mathbf{h})$ denotes the power allocated to $\eta_i$ according to $\mathcal{P}_{1}^*$,  $\mathcal{P}_{\text{WF},i}\equiv \mathcal{P}_{\text{WF}}(\eta_i)$, $V_i^*\triangleq V(\eta_i^2\mathcal{P}_{1,i}^*)$, $V_i\triangleq V(\eta_i^2\mathcal{P}_{\text{WF}}(\eta_i)) $ and $V'_i\triangleq {\sigma_N^2\eta_i^2}\big/{(\sigma_N^2+\eta_i^2\mathcal{P}_{\text{WF}}(\eta_i))^3}$. Thus, $\nu'-\nu_1 $ does not exceed $\sum\limits_{i=1}^{|\mathcal{H}|}(t_i-q_i)V_i +t_i[{\mathcal{P}_{1,i}^*-\mathcal{P}_{\text{WF},i}}]V'_i.$ Since we would like to upper bound $(\nu'-\nu_1)$ under the event ${\mathcal{T}_B}$, noting the fact that $V_i\leq 1$, we obtain that under $\mathcal{T}_B$, $$\nu'-\nu_1 \leq |\mathcal{H}|\sqrt{\frac{\log B}{B}}+\sum\limits_{i=1}^{|\mathcal{H}|}t_i[{\mathcal{P}_1^*(i,\mathbf{h})-\mathcal{P}_{\text{WF}}(\eta_i)}]V'_i.$$ Now, for $\delta_B=\sqrt{{(\log B})/{B}}$, let $\mathcal{E}\triangleq \{\mathbf{h}: \mathcal{P}_{1,i}^*-\mathcal{P}_{\text{WF},i}<\delta_B\}$. Then,  note that the summation in the above inequality can be upper bounded in the following way:
\begin{flalign}
\label{taylor_bnd1}
&\sum\limits_{i=1}^{|\mathcal{H}|}t_i[\mathcal{P}_{1,i}^*-\mathcal{P}_{\text{WF},i}]V'_i\mathbf{1}_{\mathcal{E}}+t_i[\mathcal{P}_{1,i}^*-\mathcal{P}_{\text{WF},i}]V'_i\mathbf{1}_{\mathcal{E}^c}\nonumber \\
&\stackrel{(a)}{\leq}\delta_B+\sum\limits_{i=1}^{|\mathcal{H}|}t_i[\mathcal{P}_{1,i}^*-\mathcal{P}_{\text{WF},i}]\mathbf{1}_{\mathcal{E}^c}\stackrel{(b)}{\leq}\delta_B+\bar{P}-\sum\limits_{i=1}^{|\mathcal{H}|}t_i\mathcal{P}_{\text{WF},i},
\end{flalign} 
where, in obtaining $(a)$, we have used the definition of $\mathcal{E}$ and that $V_i'\leq 1$. In getting $(b)$, we made use of the facts that, under $\mathcal{E}^c$, $\mathcal{P}_{1,i}^*-\mathcal{P}_{\text{WF},i}>0$, $V_i'\leq 1$ and  $\mathcal{P}_1^*$ satisfies the $\mathbf{ST}$ constraint with equality. Finally, note that, $\sum_{i=1}^{|\mathcal{H}|}t_i\mathcal{P}_{\text{WF},i}$ is bounded below by
\begin{flalign}
\label{taylor_bnd2}
 \sum\limits_{i=1}^{|\mathcal{H}|}t_i\mathcal{P}_{\text{WF},i}\mathbf{1}_{\mathcal{T}_B}&\stackrel{(a)}{\geq}\sum\limits_{i=1}^{|\mathcal{H}|}(q_i-\delta_B)\mathcal{P}_{\text{WF},i}\mathbf{1}_{{\mathcal{T}_B}}\\
 & \stackrel{(b)}{\geq}\bar{P}-\lambda\delta_B|\mathcal{H}|-\frac{\lambda|\mathcal{H}|}{B^2},
\end{flalign} 
where, $(a)$ follows from the fact, under the event $\mathcal{T}_B$, $t_i\geq q_i-\delta_B$. In $(b)$, $\lambda$ is as in (\ref{gold_eqn}) and follows from the fact that $\mathcal{P}_{\text{WF},i}\leq \lambda$ and $\mathbb{P}({\mathcal{T}_B}^c)\leq (2|\mathcal{H}|)/B^2$. Using (\ref{taylor_bnd1}) and (\ref{taylor_bnd2}), we see that $\nu'-\nu_1\leq 2(1+\lambda|\mathcal{H}|)\sqrt{{(\log B)}/{B}}$.
\par  As for obtaining a similar upper bound for $\nu_1-\nu'$, let ${V_i^*}'$ be the function obtained by replacing $\mathcal{P}_{\text{WF},i}$ with $\mathcal{P}_{1,i}^*$ in the function $V_i'$ defined earlier. Again, using Taylor's theorem, it follows that $V_i\leq V_i^*+{V_i^*}'(\mathcal{P}_{\text{WF},i}-\mathcal{P}_{1,i}^*).$ Note that, ${V_i^*}'\leq \eta_{\max}^2.$ Thus, $\nu_1-\nu' \leq |\mathcal{H}|\sqrt{\frac{\log B}{B}}+\sum\limits_{i=1}^{|\mathcal{H}|}q_i[{\mathcal{P}_{\text{WF},i}-\mathcal{P}_{1,i}^*}]{V_i^*}'$. Now, $\sum\limits_{i=1}^{|\mathcal{H}|}q_i[\mathcal{P}_{\text{WF},i}-\mathcal{P}_{1,i}^*]{V_i^*}'\mathbf{1}_{\mathcal{E}}\leq \delta_B\eta_{\max}^2$, where we have used the fact that ${V_i^*}'\leq \eta_{\max}^2$ and the definition of $\mathcal{E}$. Now, 
\begin{flalign*}
&\sum\limits_{i=1}^{|\mathcal{H}|}q_i[\mathcal{P}_{\text{WF},i}-\mathcal{P}_{1,i}^*]{V_i^*}'\mathbf{1}_{\mathcal{E}^c}\nonumber \stackrel{(a)}{\leq}\eta_{\max}^2\sum\limits_{i=1}^{|\mathcal{H}|}q_i[\mathcal{P}_{\text{WF},i}-\mathcal{P}_{1,i}^*]\\
& \stackrel{(b)}{=}\eta_{\max}^2\Big(\bar{P}-\sum\limits_{i=1}^{|\mathcal{H}|}q_i\mathcal{P}_{1,i}^*\Big) \\
&{=}\eta_{\max}^2\Big(\bar{P}-\sum\limits_{i=1}^{|\mathcal{H}|}q_i\mathcal{P}_{i,1}^*\mathbf{1}_{\mathcal{T}_B}-\sum\limits_{i=1}^{|\mathcal{H}|}q_i\mathcal{P}_{i,1}^*\mathbf{1}_{{\mathcal{T}_B}^c}\Big),
\end{flalign*}
\begin{flalign*}
 &\stackrel{(c)}{\leq}\eta_{\max}^2\Big(\bar{P}-\sum\limits_{i=1}^{|\mathcal{H}|}\big[(t_i-\delta_B)\mathcal{P}_{i,1}^*\mathbf{1}_{{\mathcal{T}_B}}-(t_i+\delta_B)\mathcal{P}_{i,1}^*\mathbf{1}_{{\mathcal{T}_B}^c}\big]\Big)\\
 &\stackrel{(d)}{\leq}\eta_{\max}^2\sum\limits_{i=1}^{|\mathcal{H}|}\delta_B\mathcal{P}_{i,1}^*\mathbf{1}_{{\mathcal{T}_B}}\stackrel{(e)}{\leq}\frac{\delta_B\bar{P}}{q_{\min}-\delta_B}\\
 &\stackrel{(f)}{\leq}\frac{\delta_B\bar{P}}{q_{\min}(1-\alpha)}=\frac{\bar{P}}{q_{\min}(1-\alpha)}\sqrt{\frac{\log B}{B}},
\end{flalign*}
where, in $(a)$ we have used the fact that ${V_i^*}'\leq \eta_{\max}^2$ and the definition of $\mathcal{E}$, $(b)$ follows from the fact that $\mathbb{E}_{\text{H}}[\mathcal{P}_{\text{WF}}(H)]=\bar{P},$ in $(c)$ we have made use of the definition of $\mathcal{T}_B$, in $(d)$, we have used the fact that $\mathcal{P}_1^*$ satisfies the $\mathbf{ST}$ constraint with equality, $(e)$ follows from the fact that, under $\mathcal{T}_B$, $\mathcal{P}_{1,i}^*\leq \bar{P}/(q_i-\delta_B)\leq  \bar{P}/(q_{\min}-\delta_B)$ and $(f)$ holds for some fixed $\alpha\in (0,1)$ and $B$ \textit{large enough} so that $\sqrt{(\log B)/B}<\alpha q_{\min}$.  Let $c_1=\max\big\{2(1+\lambda|\mathcal{H}|),{\bar{P}}/{q_{\min}(1-\alpha)}\big\}$. Thus, we have shown that $|\nu'-\nu_1|\leq c_1 \sqrt{(\log B)/B}$. Also, for $G$ as in Theorem \ref{Th_ST_LB2}, note that $\nu_1\geq \mathbb{E}_{\text{G}}\big[{G^2}/{2(1+G)^2}\big]=\mathbb{E}_{\text{H}}\big[\mathcal{P}^2_{\text{WF}}(H)/2\lambda^2\big]\geq \bar{P}^2/2\lambda^2,$ where, the first inequality follows from the definition of $\nu_1$, the equality follows by invoking (\ref{rel_WF_ineq}) and the last bound follows from Jensen's inequality. Thus, we have obtained a lower bound on $1/\nu_1$. Hence, we have obtained that $|(\nu'-\nu_1)/\nu_1|\leq c_2\sqrt{(\log B)/B},$ where $c_2=2c_1\lambda^2/\bar{P}^2.$ It then follows from the analysis in \cite{tomamichel2014second} that $\big\vert\mathbb{E}_{\mathbf{H}}\big[\Psi(\mathbf{H},\nu')\big]-\mathbb{E}_{\mathbf{H}}\big[\Psi(\mathbf{H},\nu_1)\big]\big\vert\leq C_1\frac{\log n}{n}$, for some constant $C_1>0$. 
\par Next, note that $\mu''(\mathbf{H})$ in the function $\Psi(\mathbf{H})$ is a sum of i.i.d. random variables taking values in $\mathcal{H}$. Hence, along with the above bound, directly invoking the result in \cite{tomamichel2014second} (Lemma 18), we obtain that $$\mathbb{E}_{\mathbf{H}}\big[{\Psi(\mathbf{H})}\big]\geq \Phi\Bigg(\sqrt{n}\frac{R- \mathbf{C}(\bar{P})}{\sqrt{V_{\text{BF}}'(\bar{P})}}\Bigg)-C_3\frac{\log n}{n}.$$ 
\end{proof}
\section*{Appendix F}
\begin{lem}
\label{lem_Berry}
Given a realization of the fading vector $\mathbf{h}$ and a  vector $\mathbf{p}\triangleq (p_1,\hdots,p_B)$ (where $p_b$ corresponds to the power allocated to the state $\eta_{i_b}$ which occurs in block $b$), let $\mathcal{B}_1(\mathbf{h},\mathbf{p})\triangleq (B\bar{P}C_{\min})^{-1}\sum\limits_{b=1}^{B}\big[K_3p_b^3+K_2p_b^2+K_1p_b+K_0\big],$ for some constants $C_{\min},~K_i>0,~i=0,\hdots,3$, not depending on $\mathbf{h}$. Let $\widehat{\mathcal{B}}_1(\mathbf{p})\triangleq \mathbb{E}_{\mathbf{H}}\big[\mathcal{B}_1(\mathbf{H},\mathbf{p})\big]$ and $\mathcal{P}_2^*$ be as defined in (\ref{prob_LB_seq2}). Then, for $\widehat{\mathcal{B}}_1\equiv \widehat{\mathcal{B}}_1(\mathcal{P}_2^*),$ $\widehat{\mathcal{B}}_1 =O\big(\frac{1}{\sqrt{Bn_c}}\big)$. 
\end{lem}
\begin{proof}
Since we are interested in the asymptotic as a function of $B$, we will prove the case for $n_c=1$ (so as to avoid unnecessary notation cluttering). The result for $n_c>1$ can be proved in an identical way. First, we will consider the term $\nu(\mathbf{h})=\sum\limits_{b=1}^{B}\nu_{h_b}(\mathbf{h})$, where $\nu_{h_b}(\mathbf{h})$ as defined in (\ref{variance_rel}). Note that
\begin{flalign*}
\nu(\mathbf{h})\geq &\sum_{b=1}^{B}\frac{\sigma_N^2h_b^2\mathcal{P}_{2,h_b}^*(\mathbf{h})}{\big(\sigma_N^2+g_b^2\big)^2} = \sum_{i=1}^{|\mathcal{H}|}B_i(\mathbf{h})\frac{\sigma_N^2\eta_i^2\mathcal{P}_{2,h_b}^*(\mathbf{h})}{\big(\sigma_N^2+\eta_i^2\mathcal{P}_{2,h_b}^*(\eta_i)\big)^2}\\
\stackrel{(a)}{\geq}& \sum_{i=1}^{|\mathcal{H}|}\frac{B_i(\mathbf{h})}{B}\mathcal{P}_{2,h_b}^*(\mathbf{h})C_{\text{min}}=B\bar{P}C_{\text{min}},
\end{flalign*}
where, in $(a)$, $C_{\text{min}}=\min\limits_{1\leq i\leq|\mathcal{H}|}\frac{\sigma_N^2\eta_i^2}{\big(\sigma_N^2+\eta_i^2\mathcal{P}_{\text{WF}}(\eta_i)\big)^2}$. Also, note that $C_{\text{min}}>0$, as $\eta_i>0$, for all $i$. Thus $\frac{1}{\nu\big(\mathbf{H}\big)}\leq \frac{1}{B\bar{P}C_{\text{min}}}$ a.s..
\par Now, we consider the term $\tau(\mathbf{h})$ as defined in (\ref{Berry}). Note that $|W_b|=$
\begin{flalign*}
&\Big\vert W_b''\big(\mathcal{P}_{2,h_b}^*(\mathbf{h})\big)-\mu_b\Big\vert \leq \Bigg\vert\frac{2h_b\sqrt{\mathcal{P}_{2,h_b}^*(\mathbf{h})}\sigma_N^2Z_b-g_b^2(Z_b^2-1)}{2\sigma_N^2\big(\sigma_N^2+g_b^2\big)}\Bigg\vert\\
&\leq \big\vert Z_b\big\vert\frac{h_b\sqrt{\mathcal{P}_{2,h_b}^*(\mathbf{h})}}{\big(\sigma_N^2+g_b^2\big)}+\big\vert (Z_b^2-1)\big\vert \frac{g_b^2}{2\sigma_N^2\big(\sigma_N^2+g_b^2\big)}
\end{flalign*}
\begin{flalign*}
&\leq \big\vert Z_b\big\vert\frac{h_b^2\mathcal{P}_{2,h_b}^*(\mathbf{h})}{\big(\sigma_N^2+g_b^2\big)}+\frac{\big\vert Z_b\big\vert}{\big(\sigma_N^2+g_b^2\big)}+ \frac{\big\vert (Z_b^2-1)\big\vert}{2\sigma_N^2}\\
&\stackrel{(a)}{\leq} \big\vert Z_b\big\vert\frac{\mathcal{P}_{2,h_b}^*(\mathbf{h})}{\lambda}+\frac{\big\vert Z_b\big\vert}{\sigma_N^2}+ \frac{\big\vert (Z_b^2-1)\big\vert}{\sigma_N^2},
\end{flalign*}
where, in $(a)$ we have made use of (\ref{rel_WF_ineq}) to get the first term.
 Thus, it can be seen that $\mathbb{E}_{\text{Z}}\big[|W_b|^3\big]\leq K_3\big[\mathcal{P}_{2,h_b}^*(\mathbf{h})\big]^3+K_2\big[\mathcal{P}_{2,h_b}^*(\mathbf{h})\big]^2+K_1\big[\mathcal{P}_{2,h_b}^*(\mathbf{h})\big]+K_0,$ for some constants $K_i>0$, not depending on $\mathbf{h}$. Next, we will show that  $\mathbb{E}_{\mathbf{H}}\sum\limits_{b=1}^{B}\big[\mathcal{P}_{2,h_b}^*(\mathbf{H})\big]^3$ is $O(B)$. To that end, since $\mathcal{P}_2^*(\mathbf{h})\in\Pi_\mathbf{h}$, 
 \begin{flalign*}
 \mathop{\overline{\lim}}_{{B\rightarrow \infty}}\sum\limits_{i=1}^{|\mathcal{H}|}t_i(\mathbf{h})\mathcal{P}_{2,i}^*(\mathbf{h})=\sum\limits_{i=1}^{|\mathcal{H}|}q_i\mathop{\overline{\lim}}_{{B\rightarrow \infty}}\mathcal{P}_{2,i}^*(\mathbf{h})=\bar{P}
 \end{flalign*}
 This implies that $$\mathop{\overline{\lim}}_{{B\rightarrow \infty}}\mathcal{P}_{2,i}^*(\mathbf{h})<{\bar{P}}/{q_i},~\text{for all}~i.$$ Hence,
 \begin{flalign}
 \hspace{-14pt}
 \label{limit_P}
 \mathop{\overline{\lim}}_{{B\rightarrow \infty}}\frac{1}{B}\sum\limits_{b=1}^{B}\big[\mathcal{P}_{2,h_b}^*(\mathbf{h})\big]^3=  \sum\limits_{i=1}^{|\mathcal{H}|}q_i\big[\mathop{\overline{\lim}}_{{B\rightarrow \infty}}\mathcal{P}_{2,i}^*(\mathbf{h})\big]^3= C_1.
 \end{flalign}
 Here, $C_1$ is a constant less than ${q^{-1}_{\min}}|\mathcal{H}|\bar{P}^2$, where $q_{\min}$ is as defined in Section \ref{mod_notn}. Next, denote $S_\delta\triangleq \bigcup\limits_{i=1}^{|\mathcal{H}|}\big\{\mathbf{h}:|t_i(\mathbf{h})-q_i|<\delta\big\}$, for $0<\delta<q_{\min}$. Also, for any power allocation vector $\mathcal{P}(\mathbf{H})$, let $\mathbf{\pi}\big(\mathbf{H},\mathcal{P}(\mathbf{H})\big)\triangleq \sum\limits_{i=1}^{|\mathcal{H}|}t_i(\mathbf{H})\mathcal{P}(\mathbf{H})$. Note that 
 \begin{flalign*}
 \mathbf{\pi}\big(\mathbf{H},\mathcal{P}_2^*(\mathbf{H})\big)= \mathbf{\pi}\big(\mathbf{H},\mathcal{P}_2^*(\mathbf{H})\mathbf{1}_{\{t_i(\mathbf{H})>0\}}\big)=\bar{P}.
 \end{flalign*}
 Hence, $\mathcal{P}_2^*(\mathbf{H})\mathbf{1}_{\{t_i(\mathbf{H})>0\}}\leq \bar{P},$ for all $i\in\{1,\hdots,|\mathcal{H}|\}$. Since, $n_i(\mathbf{H})\geq 1$ under the event $\{t_i(\mathbf{H})>0\}$, we obtain 
 \begin{equation}
 \label{pow_inequality}
 \mathcal{P}_2^*(\mathbf{H})\mathbf{1}_{\{t_i(\mathbf{H})>0\}}\leq n\bar{P}.
 \end{equation} 
 Using (\ref{pow_inequality}), $\mathbb{E}_{\mathbf{H}}\Bigg[\pi\Big(\mathbf{H},\big[\mathcal{P}_2^*(\mathbf{H})\mathbf{1}_{\{t_i(\mathbf{H})>0\}}\big]^3\Big)\mathbf{1}_{S_\delta}\Bigg]$ is upper bounded by  $\mathbb{E}_{\mathbf{H}}\Bigg[\pi\Big(\mathbf{H},B^3\bar{P}^3\Big)\mathbf{1}_{S_\delta}\Bigg]$. But, $\pi\Big(\mathbf{H},B^3\bar{P}^3\Big)=B^3\bar{P}^3\mathbb{P}\big(S_\delta\big)$. By using the union bound and Hoeffding's inequality (\cite{boucheron2013concentration}), we have
 \begin{flalign*}
 \mathbb{P}\big(S_\delta\big) \leq \sum_{i=1}^{|\mathcal{H}|}\mathbb{P}\big(|t_i(\mathbf{H})-q_i|>\delta\big)\leq 2|\mathcal{H}|e^{-2\delta B}.
 \end{flalign*}
 Thus, with $C_2\triangleq 2|\mathcal{H}|\bar{P}^3$, 
 \begin{equation}
 \label{expected_1}
 \mathbb{E}_{\mathbf{H}}\Bigg[\pi\Big(\mathbf{H},\big[\mathcal{P}_2^*(\mathbf{H})\mathbf{1}_{\{t_i(\mathbf{H})>0\}}\big]^3\Big)\mathbf{1}_{S_\delta}\Bigg]\leq C_2B^3e^{-2\delta B}.
 \end{equation}
  Next, from the fact that $\mathcal{P}_2^*(\mathbf{H})\in\Pi_\mathbf{h}$, note that, $\pi\Big(\mathbf{H},\big[\mathcal{P}_2^*(\mathbf{H})\mathbf{1}_{\{t_i(\mathbf{H})>0\}}\big]^3\Big)\mathbf{1}_{S_\delta^c}\leq \frac{\bar{P}^3}{q_{\min}-\delta}$. Then, using (\ref{limit_P}) and bounded convergence theorem, it follows that 
  \begin{equation}
 \label{expected_2}
 \mathbb{E}_{\mathbf{H}}\Bigg[\pi\Big(\mathbf{H},\big[\mathcal{P}_2^*(\mathbf{H})\mathbf{1}_{\{t_i(\mathbf{H})>0\}}\big]^3\Big)\mathbf{1}_{S_\delta^c}\Bigg]\leq C_1.
 \end{equation}
 From (\ref{expected_1}) and (\ref{expected_2}), it follows that $\mathbb{E}_{\mathbf{H}}\sum\limits_{b=1}^{B}\big[\mathcal{P}_{2,h_b}^*(\mathbf{H})\big]^3$ is $O(B)$. Hence, $\mathbb{E}_{\mathbf{H}}\sum\limits_{b=1}^{B}\big[\mathcal{P}_{2,h_b}^*(\mathbf{H})\big]^i$, $i=1,2$ are $O(B)$ as well.  Combining this with the bound on $\nu(\mathbf{h})$, we get the required result.
\end{proof}

\bibliographystyle{IEEEtran}
\bibliography{Delay_PowControl} 
 \end{document}